%% file: main.tex
\documentclass[12pt]{article}
\usepackage{bqf-paper}
\usepackage{quiver}
\usepackage{algorithm}
\usepackage{algpseudocode}
\usepackage{amsmath}
\usepackage[utf8]{inputenc} 
\usepackage[T1]{fontenc}    
\usepackage{hyperref}       
\usepackage{url}            
\usepackage{booktabs}       
\usepackage{amsfonts}       
\usepackage{nicefrac}       
\usepackage{microtype}      
\usepackage{xcolor}         
\usepackage{tikz}
\usetikzlibrary{arrows.meta,positioning,calc,decorations.pathreplacing,shapes.geometric}
\usepackage[table]{xcolor}
\usepackage{listings}

\newcommand{\sgn}{\operatorname{sgn}}

\newcommand{\D}{\mathbb D}

\newcommand{\arsinh}{\operatorname{arsinh}}

\newcommand{\simplex}{\Delta}

\newcommand{\T}{\mathbb T}
\newcommand{\K}{\mathcal K}

\newcommand{\ip}[2]{\langle #1,#2\rangle}

\newcommand{\ignore}[1]{{}}

\title{\paperTitleSmallCaps{
The Grothendieck Constant is Less Than
\(
\frac{\pi}{2\log(1+\sqrt2)}
- {\scriptstyle 10^{-5}}
\)
}}

\author{
  \normalsize
  \setlength{\tabcolsep}{15pt}
  \begin{tabular}{ccc}
    Alan Li\thanks{Equal contribution.} & Rahul Saha\footnotemark[1] & Anton Xue \\
    \small\texttt{alanli@cs.utexas.edu} & \small\texttt{rahul.saha@utexas.edu} & \small\texttt{anton.xue@utexas.edu} \\[3ex]
    Swarat Chaudhuri & Adam Klivans & Pravesh K. Kothari \\
    \small \texttt{swarat@cs.utexas.edu} & \small\texttt{klivans@cs.utexas.edu} & \small\texttt{kothari@cs.princeton.edu} \\[3ex]
    \multicolumn{3}{c}{Raghu Meka} \\
    \multicolumn{3}{c}{\small\texttt{raghum@cs.ucla.edu }}
  \end{tabular}
}

\date{}
\begin{document}


\maketitle

\input{paper/abstract}

\input{paper/introduction}

\input{paper/krivine_schemes}

\input{paper/coefficient_computation_formulas}

\input{paper/bounding_the_remainder}

\input{paper/head_computation}

\input{paper/generalization}

\input{paper/ai_guided_search}

\input{paper/tiger_partitions}

\input{paper/koenig}

\input{paper/discussion}

\input{paper/acknowledgements}

\bibliographystyle{amsalpha}
\bibliography{references}

\input{paper/appendix}

\end{document}

%% file: paper/abstract.tex
\abstract{
\noindent

 We prove that the Grothendieck constant $K_G < \frac{\pi}{2 \log (1 + \sqrt{2})} - 10^{-5}$. This improves on the work of \cite{braverman2011grothendieckconstantstrictlysmaller} who showed that $K_G < \frac{\pi}{2 \log(1 + \sqrt{2})} - \varepsilon$ for some non-explicit $\varepsilon > 0$.

\ignore{
A landmark result was due to Krivine in 1977, who proved that \(K_G \leq \tfrac{\pi}{2 \log (1 + \sqrt{2})}\) via the analysis of hyperplane rounding schemes.
A rounding scheme is used to discretize continuous vectors, and improving the upper bound on \(K_G\) can be reduced to the discovery and analysis of new schemes.
Krivine conjectured that this bound was optimal --- a widely held belief for over thirty years until its refutation by Braverman-Makarychev-Makarychev-Naor (BMMN) in 2011.
Specifically, BMMN used a technique called mixing to analyze a new class of rounding schemes, and showed that \(K_G < \tfrac{\pi}{2 \log (1 + \sqrt{2})} - \varepsilon\) for some \(\varepsilon > 0\).
However, BMMN's mixing provides only a partial understanding of rounding schemes, and does not translate to strong numerical improvements: extending their analysis only gives an improvement on the order of \(\varepsilon \approx 10^{-53}\).

In this work, we overcome the limitations of BMMN's mixing argument through direct analytic control of the rounding and its associated local inverse. This allows us to better analyze rounding schemes and achieve a tighter bound on $K_G$.
Building on this approach, we automate the discovery of schemes through an AI-guided search, through which we prove that
\[
K_G < \frac{\pi}{2 \log (1 + \sqrt{2})} - 10^{-5}.
\]
Furthermore, we establish asymptotic results on the K\"onig bilinear form, an object closely related to $K_G$.
In summary, this work establishes the best-known upper bound on the Grothendieck constant and improves our understanding of rounding schemes.}}
    
\medskip

\newpage 
\tableofcontents
\newpage

%% file: paper/introduction.tex
\section{Introduction}


The Grothendieck inequality~\cite{Gro53} from functional analysis shows that~\cite{LP68}, there exists a universal constant \(K < \infty\) such 
that for every \(m,n \in \N\) and every matrix \(A=(a_{ij}) \in \R^{m \times n}\),
\begin{equation}
\label{eqn:KG_inequality}
    \max_{\varepsilon_i,\delta_j \in \{\pm 1\}}
    \bigg|
    \sum_{i=1}^m \sum_{j=1}^n a_{ij}\varepsilon_i\delta_j
    \bigg|
    \leq 
    \sup_{\substack{u_i,v_j \in S^{m+n-1}}}
    \bigg|
    \sum_{i=1}^m \sum_{j=1}^n a_{ij}\angles{u_i,v_j}
    \bigg|
    \leq 
    K
    \max_{\varepsilon_i,\delta_j \in \{\pm 1\}}
    \bigg|
    \sum_{i=1}^m \sum_{j=1}^n a_{ij}\varepsilon_i\delta_j
    \bigg|.
\end{equation}
Here \(S^{m+n-1}\) denotes the unit sphere in \(\R^{m+n}\).
The Grothendieck constant $K_G$ is the infimum of \(K\) for which this inequality holds.

The Grothendieck inequality is equivalent to the statement that the integrality gap of the canonical SDP relaxation for maximizing the bilinear form $x^{\top} Ay$ over $x \in \{-1,1\}^m$ and $y \in \{-1,1\}^n$ is an absolute constant $K_G$. As a result, one obtains efficient algorithms for approximating the matrix \emph{cut norms} that, among other applications, lead to algorithms for finding Szemer\'{e}di partitions of graphs~\cite{AN06}. Beyond this natural role in the design of approximation algorithms in computer science, the inequality is important in several subfields, including the geometry of Banach spaces, $ C^*$-algebras, harmonic analysis, operator spaces, and quantum mechanics (see the monographs~\cite{lindenstrauss2013classical,  pisier2011grothendieckstheorempastpresent, khot2011grothendiecktypeinequalitiescombinatorialoptimization} for more).



Despite a long line of work on the problem, the exact value of \(K_G\) remains unknown. Grothendieck proved that \(K_G \leq \sinh(\tfrac{\pi}{2}) = 2.3012 \ldots\). The next advance was made in 1977 by Krivine~\cite{krivine1977constante}, who  proved that
\begin{equation}
    K_G \leq \tfrac{\pi}{2 \log(1 + \sqrt{2})} = 1.7822 \ldots
\end{equation}
Krivine's proof involved developing an explicit \textit{rounding scheme}, an algorithm that maps an input set of vectors \(u_i, v_j \in S^{d-1}\)
into signs
\(\varepsilon_i, \delta_j \in \{\pm 1\}\).
Krivine's scheme consisted of first lifting \(\{u_i, v_j\}\) to an infinite-dimensional space and then partitioning them by a random hyperplane into signed \(\{\varepsilon_i, \delta_j\}\). He showed that this hyperplane partition yields the above bound and conjectured that it is optimal, leading to subsequent efforts to find matching lower bounds~\cite{Davie1984LowerBoundKG, Reeds1991LowerBoundKG}, with notable recent efforts being~\cite{heilman2026lowerboundgrothendiecksconstant, jones2026grothendieckconstantstrictlylarger}.

Krivine's conjecture was refuted three decades later by Braverman,
Makarychev, Makarychev, and Naor~\cite{braverman2011grothendieckconstantstrictlysmaller}, who showed that, in fact, \(K_G < \tfrac{\pi}{2 \log (1 + \sqrt{2})} - \varepsilon\) for some \(\varepsilon > 0\).
Their proof relied on a ``mixing'' of two rounding strategies, i.e., applying a new non-hyperplane scheme with a small probability and the hyperplane rounding scheme with the rest. 
This idea of mixing is natural in the context of analyzing rounding algorithms and is informally analogous to first-order methods that obtain a better guess for an optimum by exploiting a nonzero gradient. Such a first-order analysis suffices to prove the \textit{existence} of a scheme that outperforms Krivine's rounding. 
However, this technique makes it difficult to derive an explicit improvement. 

In this work, we develop new rounding techniques to derive an explicit improvement to the Grothendieck constant. 

\begin{theorem}[Main]
\label{thm:explicit-improvement}
The Grothendieck constant $K_G$ satisfies:
\[
    K_G
    \leq
    \frac{\pi}{2\log(1+\sqrt2)}
    -
    6.039\cdot 10^{-5}
\]
\end{theorem}

Our main idea for achieving this improvement is a new technique for analyzing rounding schemes that goes beyond the first-order analysis in~\cite{braverman2011grothendieckconstantstrictlysmaller}. As a result, a \emph{single} simple new scheme already suffices to obtain a non-trivial improvement on $K_G$:

\begin{theorem}
\label{thm:pure-h3}
Let \(\Hermite(x)=(x^3-3x)/\sqrt{6}\) be the third-degree orthonormal probabilist's Hermite polynomial, and set
\[
 \tau=0.01825,
 \qquad
 \eta= \sqrt{6}\tau
 =0.0447031878\ldots .
\]
Define odd partitions \(f,g:\R^2\to\{\pm1\}\) by
\[
    f(x,y)=\sgn(y-\eta \Hermite(x)),
    \qquad
    g(x,y)=\sgn(y+\eta \Hermite(x)).
\]
The rounding scheme generated by this pair of partitions gives the upper bound,
\[
    K_G
    <
   \frac{\pi}{2\log(1+\sqrt2)} - 
    9.7\cdot 10^{-6}
\]
\end{theorem}

Our main result above, elaborated in \cref{thm:main-explicit}, is then obtained via a more sophisticated version of the space partition above \emph{mixed} with the classical hyperplane rounding. This partitioning scheme is described by  a 
degree-\(9\) Hermite-threshold rounding of the form
\[
    f_\Phi(x_1,x_2)=\sgn(x_2-P(x_1)),
    \qquad
    g_\Phi(x_1,x_2)=\sgn(x_2-Q(x_1)),
\]
where \(P,Q:\mathbb R\to\mathbb R\) are polynomials. This scheme and the appropriate mixing weights were discovered by an AI-guided search. 

This AI-guided search uses AI agents to iteratively explore the space of partitions and propose strategies for parameter optimization.
Crucially, this allows us to automate the discovery and analysis of rounding schemes, which were previously laboriously done by hand.
We developed this search framework as part of a broader project (under submission) and apply it here to the Grothendieck constant.

\paragraph{K\"onig's bilinear form:} Our methods also allow us to obtain space partitions with improved K\"onig's value \(\mathcal{B}_K (f, g)\) (see Definition~\ref{def:konig} for formal definition). The K\"onig value is a bilinear form on pairs of space partitions \(f, g: \mathbb{R}^d \to \{\pm 1\}\) and was introduced by K\"onig in connection with questions of unconditional convergence, and later appeared as an analytic proxy to study Krivine's conjecture, following the work of Haagerup, K\"onig, and Tomczak-Jaegermann~\cite{Kon00, haagerup1987new}. In particular, partitions with good K\"onig's value are natural candidates for usage in rounding schemes for the Grothendieck inequality. Indeed, this is the proof strategy employed in the work of ~\cite{braverman2011grothendieckconstantstrictlysmaller}. 

Our main result improves on their work and shows: 
\begin{theorem}\label{th:konigintro}
There is an explicit family of partition pairs
\(
    f_d,g_d:\mathbb R^{d+2}\to\{\pm1\}
\)
whose normalized K\"onig values satisfy
\[
    \mathcal{B}_K(f_d,g_d)=C+O(d^{-1/3})
\]
for a constant satisfying
\(C > 0.59357 \ldots\). Notably, $C$ is a value obtained by optimizing a particular expression. 
\end{theorem}
This improves on the hyperplane benchmark of $\frac{2}{\pi}\log(1+\sqrt2)=0.56109 \ldots$. Our approach also provides numerical evidence against a conjecture of ~\cite{braverman2011grothendieckconstantstrictlysmaller}, who conjectured that a class of space partitions called  ``tiger partitions'' are optimal schemes for \(K_G\). In particular, the numerical evidence we provide shows that while the tiger partitions attain a better \(\mathcal{B}_K\) value than the Krivine hyperplane baseline, they produce a worse upper bound on \(K_G\).

This paper is organized as follows.

\begin{enumerate}
    \item \cref{sec:krivine-rounding-scheme} recalls Krivine rounding schemes. 
    \item  \cref{sec:coefficient-computation-formulas} describes preliminary analytic tools that are then implemented in \texttt{Arb}, a C library that can provably track precision and error in numerical computations.
    \item \cref{first improvement for grothendieck constant}
    proves \cref{thm:pure-h3}. Then, this approach is generalized in \cref{sec:generalization}.
    \item 
    \cref{sec:generalization} proves the explicit \(10^{-5}\) improvement, and  \cref{sec:explicit-improvement} describes the AI-assisted search over Hermite-threshold partitions that led to the improvement.
     
    \item \cref{sec:tiger-partitions} introduces K\"onig's bilinear form, and provides numerical evidence against a conjecture in~\cite{braverman2011grothendieckconstantstrictlysmaller} on the extremizers of the form known as ``tiger partitions''. Finally,  \cref{sec:konig-not-sufficient} proves an asymptotic result for K\"onig's bilinear form.
\end{enumerate}

%% file: paper/krivine_schemes.tex
\section{Krivine Rounding Scheme}
\label{sec:krivine-rounding-scheme}

In this section, we recall Krivine schemes and their properties, as described in ~\cite{krivine1977constante, braverman2011grothendieckconstantstrictlysmaller, naor2014krivine}. A rounding scheme is a pair of partitions of $\R^d$, $f,g: \R^d \rightarrow \{-1,1\}$, such that $f,g$ are odd functions. We will work with $f,g$ that satisfy some additional properties. To describe these, we first introduce some basic setup~\cite{braverman2011grothendieckconstantstrictlysmaller}. 


\begin{definition}[Gaussian correlation function]\label{def:gaussian-correlation}
Fix $k \in \mathbb{N}$, and let $f,g:\mathbb{R}^k\to\{-1,1\}$ be odd measurable functions.
Let $G_1,G_2$ be independent standard Gaussian vectors in $\mathbb{R}^k$.

The \emph{arcsine-normalized correlation function} associated to $(f,g)$ is
\[
H_{f,g}(t)
:=
\frac{\pi}{2}\,
\mathbb{E}\left[
f\left(G_1\right)
g\left(
tG_1+\sqrt{1-t^2}G_2
\right)
\right],
\qquad -1<t<1.
\]
\end{definition}

We will use the following basic properties of the Gaussian correlation function and its inverse $H_{f,g}^{-1}$ in our analysis (see \cite{braverman2011grothendieckconstantstrictlysmaller}).

\begin{lemma}\label[lemma]{lemma:correlation-profile-analytic}
Let $f,g:\mathbb{R}^k\to\{-1,1\}$ be odd measurable functions. Then $H_{f,g}$ is odd and
extends holomorphically to the strip
\[
S:=\{z\in\mathbb{C}: |\operatorname{Re}z|<1\}.
\]
In particular, if $H_{f,g}'(0)\neq 0$, then $H_{f,g}$ has a unique \textbf{local inverse branch} at
the origin, i.e., there are open neighborhoods
$U,V\subset\C$ of $0$ such that
\(
    H_{f,g}:U\to V
\)
is biholomorphic. We write
\(
    H_{f,g}^{-1}:V\to U
\)
for this local inverse, chosen so that $H_{f,g}^{-1}(0)=0$. This inverse is
unique when restricted to this branch, and is an odd function.
\end{lemma}


\begin{lemma}[Local Inverse Expansion]\label[lemma]{lemma:local-inverse-expansion}

Let $f,g:\mathbb{R}^k\to\{-1,1\}$ be odd and measurable functions with $H'_{f,g}(0) \neq 0$. Let $H^{-1}_{f,g}$ be the unique local inverse branch of $H_{f,g}$ at the origin. Then,
\begin{itemize}
\item $H^{-1}_{f,g}$ is odd. 
    \item $H^{-1}_{f,g}$ has a Taylor-series expansion at the origin: $H^{-1}_{f,g}(\zeta) = \sum_{n\geq 1} a_n \zeta^n$. 
\end{itemize}
Further, if $\gamma > 0$ is such that $\sum_{n \geq 1} |a_n| \gamma^n \leq 1$, then the Grothendieck constant 
$$ K_G \leq \frac{\pi}{2\gamma}.$$ 
\end{lemma}

For completeness, we give a proof of this  in~\cref{app:krivine-preprocessing}.

With the above notation, we define the \textbf{inverse majorant series} for the rounding scheme $(f,g)$ as 
\[
M_{f,g}(s):=\sum_{n\ge 1}|a_n|s^n.
\]

We are now ready to define Krivine schemes. 

\begin{definition}[$\gamma$-admissible Krivine Rounding Scheme]
\label[definition]{def:gamma-admissible}
    For $f,g:\mathbb{R}^k\to\{-1,1\}$ odd and measurable functions with $H'_{f,g}(0) \neq 0$, and $\gamma > 0$, we say $(f,g)$ is a $\gamma$-admissible Krivine rounding scheme if the conditions of \cref{lemma:local-inverse-expansion} hold. 
\end{definition}
We note that Krivine~\cite{Kri77} used a $\gamma$-admissible Krivine rounding scheme to show that \(K_G \le \frac{\pi}{2\gamma}.\)


%% file: paper/coefficient_computation_formulas.tex
Next, we describe the preliminary analytic backbone of the computation: the reduction from planar graph-threshold schemes to one-dimensional integrals.

\section{Preliminaries}
\label{sec:coefficient-computation-formulas}
In this section, we develop some basic calculations that will be used in the remainder of the paper. 
Our proof requires computing the Taylor series expansions of the local inverse of the Gaussian correlation function $H_{f,g}$. In this section, we record the analytic identities used to compute the Taylor coefficients of \(H_{f,g}\) and \(H_{f,g}^{-1}\).  

Consider the schemes
\[
    f_P(x,y)=\sgn(y-P(x)),
    \qquad
    g_Q(x,y)=\sgn(y-Q(x)),
\]
where \(P,Q:\R\to\R\) are odd measurable functions. The special schemes used elsewhere in the paper are obtained by substituting the corresponding choices
of \(P\) and \(Q\). 

We use standard facts about Hermite polynomials, including orthogonality, the generating function, and covariance identities. See, for
example, \cite{Andrews_Askey_Roy_1999}.  

Let $\mu$ be the standard Gaussian measure on $\R$, so that
\[
    d\mu(x):=\frac{1}{\sqrt{2\pi}}e^{-x^2/2}\,dx.
\]
Let \(\HermBasis_n\) denote the \(n\)-th orthonormal probabilist's Hermite polynomial. Then \(\{\HermBasis_n\}_{n\ge0}\) is an orthonormal basis for \(L^2(\mu)\).

\begin{lemma}[Hermite coefficient formula for \(H_{f,g}\)]
\label[lemma]{lemma:hermite-mehler-diagonalization}
Fix \(f,g:\R^2\to\{-1,1\}\), and define
\[
    \mcal A_{a,b}
    :=
    \int_{\R^2}
        f(x,y)\HermBasis_a(x)\HermBasis_b(y)
    \,d\mu(x)d\mu(y),
    \qquad
    \mcal B_{a,b}
    :=
    \int_{\R^2}
        g(x,y)\HermBasis_a(x)\HermBasis_b(y)
    \,d\mu(x)d\mu(y).
\]
Then
\[
    H_{f,g}(z)
    =
    \frac{\pi}{2}
    \sum_{a,b\ge0}
        \mcal A_{a,b}\mcal B_{a,b} z^{a+b}.
\]
Equivalently,
\[
    [z^n]H_{f,g}(z)
    =
    \frac{\pi}{2}
    \sum_{a=0}^{n}
        \mcal A_{a,n-a}\mcal B_{a,n-a}.
\]
\end{lemma}

\begin{proof}
We first prove the identity for real \(|z|<1\).  Let
\[
    (X_1,Y_1),\ (X_2,Y_2)
\]
be two independent pairs of jointly Gaussian random variables, where each
\(X_i\) and \(Y_i\) is distributed according to \(d\mu\), and where \(X_i\) and
\(Y_i\) have correlation parameter \(z\).  In this notation,
\[
    H_{f,g}(z)
    =
    \frac{\pi}{2}\,
    \E\bigl[f(X_1,X_2)g(Y_1,Y_2)\bigr].
\]

We use the standard Hermite covariance identity
\[
    \E\bigl[\HermBasis_m(X)\HermBasis_n(Y)\bigr]
    =
    \delta_{mn}z^n,
\]
whenever \(X,Y\) are \(d\mu\)-Gaussian variables with correlation parameter
\(z\).  Applying this identity independently in the two coordinates gives
\[
    \E\bigl[
        \HermBasis_a(X_1)\HermBasis_b(X_2)
        \HermBasis_c(Y_1)\HermBasis_d(Y_2)
    \bigr]
    =
    \delta_{ac}\delta_{bd}z^{a+b}.
\]

Now expand \(f\) and \(g\) in the tensor-product Hermite basis:
\[
    f(x,y)
    =
    \sum_{a,b\ge0}
        \mcal A_{a,b}\HermBasis_a(x)\HermBasis_b(y),
    \qquad
    g(x,y)
    =
    \sum_{c,d\ge0}
        \mcal B_{c,d}\HermBasis_c(x)\HermBasis_d(y).
\]
Substituting these expansions into the expectation and using the preceding
covariance identity, all off-diagonal terms vanish.  Hence, for finite Hermite expansions,
\[
\begin{aligned}
    H_{f,g}(z)
    &=
    \frac{\pi}{2}
    \sum_{a,b,c,d\ge0}
        \mcal A_{a,b}\mcal B_{c,d}
        \E\bigl[
            \HermBasis_a(X_1)\HermBasis_b(X_2)
            \HermBasis_c(Y_1)\HermBasis_d(Y_2)
        \bigr]  \\
    &=
    \frac{\pi}{2}
    \sum_{a,b\ge0}
        \mcal A_{a,b}\mcal B_{a,b}z^{a+b}.
\end{aligned}
\]
For general \(f,g\), this follows by approximation in \(L^2\).

The series is absolutely convergent for \(|z|<1\), since
\[
    \sum_{a,b\ge0}
        |\mcal A_{a,b}\mcal B_{a,b}|\,|z|^{a+b}
    \le
    \left(\sum_{a,b\ge0}\mcal A_{a,b}^2\right)^{1/2}
    \left(\sum_{a,b\ge0}\mcal B_{a,b}^2\right)^{1/2}
    =
    \|f\|_2\|g\|_2
    =
    1.
\]
Thus the identity extends to complex \(|z|<1\) by analyticity.  Finally,
collecting the terms with total degree \(a+b=n\) gives
\[
    [z^n]H_{f,g}(z)
    =
    \frac{\pi}{2}
    \sum_{a=0}^{n}
        \mcal A_{a,n-a}\mcal B_{a,n-a}.
\]
\end{proof}


\begin{lemma}
\label[lemma]{lemma:threshold-coefficient-reduction}
For the graph-threshold pair \(f_P,g_Q\), define
\[
    \ThresholdCoeff_b(u)
    :=
    \int_{\R}\sgn(y-u)\HermBasis_b(y)\,d\mu(y).
\]
Then the Hermite coefficients from
\cref{lemma:hermite-mehler-diagonalization} are
\[
    \mcal A_{a,b}
    =
    \int_{\R}
        \HermBasis_a(x)\ThresholdCoeff_b(P(x))
    \,d\mu(x),
    \qquad
    \mcal B_{a,b}
    =
    \int_{\R}
        \HermBasis_a(x)\ThresholdCoeff_b(Q(x))
    \,d\mu(x).
\]
Moreover,
\[
    \ThresholdCoeff_0(u)=-\operatorname{erf}\left(\frac{u}{\sqrt{2}}\right),
    \qquad
    \ThresholdCoeff_b(u)
    =
    \sqrt{\frac{2}{\pi}}
    e^{-u^2/2}
    \frac{\HermBasis_{b-1}(u)}{\sqrt{b}}
    \quad (b\ge1).
\]
\end{lemma}
\begin{proof}
We prove the formula for \(\mcal A_{a,b}\); the proof for
\(\mcal B_{a,b}\) is identical with \(P\) replaced by \(Q\). By Fubini's
theorem,
\[
\begin{aligned}
    \mcal A_{a,b}
    &=
    \int_{\R}
        \HermBasis_a(x)
        \left(
            \int_{\R}
                \sgn(y-P(x))\HermBasis_b(y)\,d\mu(y)
        \right)
    d\mu(x) \\
    &=
    \int_{\R}
        \HermBasis_a(x)\ThresholdCoeff_b(P(x))
    \,d\mu(x).
\end{aligned}
\]
It remains to compute \(\ThresholdCoeff_b\). For \(b=0\), since
\(\HermBasis_0=1\),
\[
    \ThresholdCoeff_0(u)
    =
    \mu((u,\infty))-\mu((-\infty,u))
    =
    1-2\mu((-\infty,u))
    =
    -\operatorname{erf}\left(\frac{u}{\sqrt{2}}\right).
\]
Now suppose \(b\ge1\). Since \(\HermBasis_b\) is orthogonal to constants,
\[
\begin{aligned}
    \ThresholdCoeff_b(u)
    &=
    \int_u^\infty \HermBasis_b(y)\,d\mu(y)
    -
    \int_{-\infty}^u \HermBasis_b(y)\,d\mu(y) \\
    &=
    -2\int_{-\infty}^u \HermBasis_b(y)\,d\mu(y).
\end{aligned}
\]
Using
\[
    \frac{d}{dy}
    \left(
        e^{-y^2/2}\HermBasis_{b-1}(y)
    \right)
    =
    -\sqrt{b}\,
    e^{-y^2/2}\HermBasis_b(y),
\]
we get
\[
\begin{aligned}
    \int_{-\infty}^u \HermBasis_b(y)\,d\mu(y)
    &=
    \frac{1}{\sqrt{2\pi}}
    \int_{-\infty}^u
        e^{-y^2/2}\HermBasis_b(y)\,dy \\
    &=
    -\frac{1}{\sqrt{2\pi}}
    e^{-u^2/2}
    \frac{\HermBasis_{b-1}(u)}{\sqrt{b}}.
\end{aligned}
\]
Therefore
\[
    \ThresholdCoeff_b(u)
    =
    \sqrt{\frac{2}{\pi}}
    e^{-u^2/2}
    \frac{\HermBasis_{b-1}(u)}{\sqrt{b}},
\]
as claimed.
\end{proof}

\begin{corollary}
\label[corollary]{cor:coefficient-formula-combined}
Combining the preceding two lemmas gives the coefficient formula
\[
    [z^n]H_{f_P,g_Q}(z)
    =
    \frac{\pi}{2}
    \sum_{a=0}^{n}
        \left(
        \int_{\R}
            \HermBasis_a(x)\ThresholdCoeff_{n-a}(P(x))
        \,d\mu(x)
        \right)
        \left(
        \int_{\R}
            \HermBasis_a(x)\ThresholdCoeff_{n-a}(Q(x))
        \,d\mu(x)
        \right).
\]
\end{corollary}

\begin{lemma}[Inverse coefficients]
\label[lemma]{lemma:lagrange-inverse-coefficients}
Suppose \(H\) is odd and write $ H(z)=z\OddFactor(z^2)$.
Then \(H^{-1}\) has only odd coefficients.  Writing
\[
    H^{-1}(z)=\sum_{m\ge1}a_m z^m,
\]
we have
\[
    a_{2k+1}
    =
    \frac{1}{2k+1}
    [t^k]\OddFactor(t)^{-(2k+1)},
    \qquad
    a_{2k}=0.
\]
\end{lemma}

\begin{proof}
This is the standard Lagrange inversion formula applied to the odd map
\(H(z)=z\OddFactor(z^2)\).  Oddness of \(H\) implies oddness of the local
inverse branch, so the even inverse coefficients vanish.
\end{proof}

%% file: paper/bounding_the_remainder.tex
\section{First Improvement for Grothendieck Constant}
\label{first improvement for grothendieck constant}

We now define our explicit scheme for proving \cref{thm:pure-h3}.

\begin{definition}[The $\Hermite$ scheme]\label[lemma]{def:h3-scheme}
Let
\[
\Hermite(x)=\frac{(x^3-3x)}{\sqrt{6}},
\qquad
\tau=0.01825,
\qquad
\eta= \sqrt{6}\tau
 =0.0447031878\ldots .
\]
Define odd measurable functions $f,g:\R^2\to\{\pm1\}$ by
\[
    f(x,y)=\sgn(y-\eta \Hermite(x)),
    \qquad
    g(x,y)=\sgn(y+\eta \Hermite(x)).
\]
We call $(f,g)$ the $\Hermite$ scheme.
\end{definition}

We will prove the following theorem (complete version of \cref{thm:pure-h3}) in the following sections:

\begin{theorem}[Improved Grothendieck Constant]\label[theorem]{thm:no-mixing-beats-hyperplane}
Let
\[
\rhostarval:=\log(1+\sqrt2),
\qquad
\gzero:=\rhostarval+4.8\cdot 10^{-6}.
\]
Then the $\Hermite$ scheme is a $\gzero$-admissible Krivine scheme.  In particular,
\[
K_G
\le \frac{\pi}{2\gzero}
<
\frac{\pi}{2\log(1+\sqrt2)}
-
9.7\cdot 10^{-6}.
\]
\end{theorem}

\paragraph{Proof Overview:} Our proof has two main components. Recall that using \cref{lemma:local-inverse-expansion}, the main goal is to show that 
$$\sum_{n \geq 1} |a_n| \gamma_0^n \leq 1,$$
where $(a_n)_{n \geq 1}$ are the Taylor-series coefficients of $H_{f,g}^{-1}$. We split this into parts. 
\begin{enumerate}
    \item An argument that bounds the \textit{remainder term}, $\sum_{n > k_0} |a_n| \gamma_0^n$  (at most $2.5 \cdot 10^{-9}$) for $k_0 = 220$. This analytic portion relies on elementary complex analysis, a carefully chosen change of coordinates, and a numerical computation that uses \texttt{Arb} \cite{johansson2017arb}. This is done in \cref{sec:bounding-remainder}. This approach can also extend to schemes similar to the $\Hermite$ scheme. See \cref{sec:generalization}.
    \item A numerical computation that uses \texttt{Arb} to compute the \emph{head}, $\sum_{n \leq k_0} |a_n| \gamma_0^n$ with rigorous enclosures showing that this value is less than $1 - 2.5 \cdot 10^{-9}$. This is done in \cref{sec:numerical-computation-of-head}.  
\end{enumerate}

\paragraph{A note on \texttt{Arb}.}
All numerical computations used in the proof are rigorously evaluated using the
\texttt{Arb} interval arithmetic library \cite{johansson2017arb}.  \texttt{Arb} represents real and
complex numbers by midpoint--radius intervals, or ``balls'', and performs
arithmetic with outward rounding, so each output ball is guaranteed to contain
the exact mathematical value. Thus the computations below are not merely
high-precision floating point computations: every Taylor coefficient, inverse
coefficient, integral, and finite sum used in the certificate is enclosed in a
rigorous interval.

Concretely, interval enclosures are propagated through elementary functions,
power series, algebraic operations, and recurrence formulas for inverse-series
coefficients, including all rounding errors.  For the integrals, the quadrature routine produces rigorous
enclosures on subintervals, equivalently using local Taylor models with explicit remainder bounds.  
Consequently, whenever we assert an inequality certified by \texttt{Arb}, the
computed interval lies entirely on the required side of the bound.  

\subsection{Bounding the Remainder}\label{sec:bounding-remainder}

We now give our main analytic argument to bound the remainder term of the inverse majorant series of $\Hshort^{-1}$ for the $\Hermite$ scheme. A similar argument extends to other schemes as well, which we later show in \cref{sec:generalization}. The rest of this section is dedicated to the following lemma. 
\begin{lemma}[Majorant Remainder Bound]\label[lemma]{lemma:tail-bound}
Let $(f,g)$ be the $\Hermite$ scheme and let  
\[
H_{f,g}^{-1}(\zeta)=\sum_{n\ge1}a_n\zeta^n,
\qquad
\gzero=\log(1+\sqrt2)+4.8\cdot10^{-6}.
\]
Then
\[
\sum_{n>220}|a_n|\gzero^n<2.5\cdot10^{-9}.
\]
\end{lemma}
For brevity, we assume $(f,g)$ is fixed to be the $\Hermite$ scheme and we write $\Hshort^{-1}$, $\Mshort$ instead of $H^{-1}_{f,g}$, $M_{f,g}$. First, we prove a bound for $H$ in a strip. This is an analog of the strip estimate used in \cite[Theorem~4.2]{braverman2011grothendieckconstantstrictlysmaller}.

\begin{lemma}[Strip bound]\label[lemma]{lemma:strip-bound}
Let $f,g:\R^2\to\{-1,1\}$ be measurable and let $H=H_{f,g}$ be the
arcsine-normalized correlation function.
Then $H$ extends holomorphically to
\[
 S=\{z\in\C: |\operatorname{Re}z|<1\}.
\]
Moreover, if $z=a+ib\in S$, then
\[
 |H(z)|\le \frac{\pi}{2}\frac{|1-z^2|}{1-a^2}.
\]
\end{lemma}

\begin{proof}
For real $|t|<1$, the arcsine-normalized correlation admits the integral representation
\[
H(t)=\frac{1}{8\pi(1-t^2)}
\int_{\R^2}\int_{\R^2}
f(x)g(y)
\exp\left(
-\frac{\|x\|^2+\|y\|^2-2t\ip{x}{y}}{2(1-t^2)}
\right)\,dx\,dy .
\]
For $z\in S$, this formula defines the candidate holomorphic extension. Since $|f|=|g|=1$,
\[
|H(z)|\le
\frac{1}{8\pi |1-z^2|}
\int_{\R^2}\int_{\R^2}
\exp\left(
-\operatorname{Re}\frac{\|x\|^2+\|y\|^2-2z\ip{x}{y}}{2(1-z^2)}
\right)\,dx\,dy .
\]
Use the identity
\[
\frac{\|x\|^2+\|y\|^2-2z\ip{x}{y}}{2(1-z^2)}
=
\frac14\frac{\|x+y\|^2}{1+z}
+
\frac14\frac{\|x-y\|^2}{1-z}.
\]
Writing $z=a+ib$, set
\[
A_+ := \operatorname{Re}\frac{1}{1+z}
=\frac{1+a}{(1+a)^2+b^2},
\qquad
A_- := \operatorname{Re}\frac{1}{1-z}
=\frac{1-a}{(1-a)^2+b^2}.
\]
Both are positive for $|a|<1$. With $u=x+y$ and $v=x-y$, the Jacobian from $(u,v)$ to $(x,y)$ in $\R^4$ is $1/4$. Thus the Gaussian integral is
\[
\frac14\cdot\frac{4\pi}{A_+}\cdot\frac{4\pi}{A_-}
=\frac{4\pi^2}{A_+A_-}.
\]
Finally,
\[
A_+A_-=
\frac{1-a^2}{|1+z|^2|1-z|^2}
=
\frac{1-a^2}{|1-z^2|^2}.
\]
Substitution gives
\[
|H(z)|\le
\frac{1}{8\pi |1-z^2|}\cdot
\frac{4\pi^2 |1-z^2|^2}{1-a^2}
=
\frac{\pi}{2}\frac{|1-z^2|}{1-a^2}.
\]
The same estimate is uniform on compact subsets of $S$, so the integral converges locally uniformly and the extension is holomorphic by dominated convergence.
\end{proof}

Next, we make a key change of variable that gives us analytic control over the local inverse. 

\begin{definition}[Key variable change]\label[lemma]{def:key-change-of-variable}
For the rest of the proof, put
\[
F(w)\defeq H(\sin w).
\]
\end{definition}

For the hyperplane scheme $H_0(t)=\arcsin t$, hence $H_0(\sin w)=w$ on the principal strip. This is why the comparison below is made against the identity map. 

\begin{lemma}\label[lemma]{lemma:sine-coordinate-bound}
For $F(w)=H(\sin w)$, the function $F$ is holomorphic on $|w|<1.1$ and
\[
 |F(w)-w|<15.3\qquad (|w|\le1.1).
\]
\end{lemma}

\begin{proof}
Write $w=u+iv$. Then
\[
\operatorname{Re}(\sin w)=\sin u\cosh v.
\]
By \cref{lemma:disk-to-strip-check}, if $u^2+v^2\le 1.1^2$, then
\[
|\sin u|\cosh v<0.946<1.
\]
Therefore $\sin w\in S$ for $|w|\le 1.1$, and \cref{lemma:strip-bound} implies that $F$ is holomorphic on $|w|<1.1$.  Applying \cref{lemma:strip-bound} with $z=\sin w$ gives
\[
|F(w)|
\le
\frac{\pi}{2}\frac{|1-\sin^2 w|}{1-(\operatorname{Re}\sin w)^2}
=
\frac{\pi}{2}\frac{|\cos w|^2}{1-(\operatorname{Re}\sin w)^2}.
\]
By \cref{lemma:cosine-ratio-check}, the last expression is $<14.19$ on $|w|\le 1.1$.  Hence
\[
|F(w)-w|\le |F(w)|+|w|<14.19+1.1<15.3.
\]
\end{proof}

Next, we prove a general sufficient condition for the radius of convergence and supremum of the inverse of a holomorphic function on a disk. 

\begin{lemma}[Rouch\'e inverse disk]\label[lemma]{lemma:rouche-inverse-disk}
Let $F$ be holomorphic on $|w|<r$, with $F(0)=0$, and suppose that for some
$R<M<r$,
\[
 \sup_{|w|=M}|F(w)-w|<M-R.
\]
Then:
\begin{enumerate}[label=(\roman*)]
    \item $F^{-1}$ is holomorphic on $R\D$;
    \item for $|\zeta|<R$, the inverse satisfies $|F^{-1}(\zeta)|<M$.
\end{enumerate}
\end{lemma}

\begin{proof}
Fix $|\zeta|<R$.  On $|w|=M$,
\[
|(F(w)-\zeta)-(w-\zeta)|=|F(w)-w|<M-R<M-|\zeta|\le |w-\zeta|.
\]
By Rouch\'e's theorem, $F(w)-\zeta$ and $w-\zeta$ have the same number of zeros in $|w|<M$, counted with multiplicity.  Since $w-\zeta$ has exactly one zero in this disk, $F(w)=\zeta$ has exactly one solution in $|w|<M$, counted with multiplicity.  In particular, the solution is simple.  The holomorphic inverse function theorem gives local inverse branches, and uniqueness makes these branches glue over $R\D$.  The same argument gives $|F^{-1}(\zeta)|<M$.
\end{proof}

\begin{corollary}[Radius and supremum of $|F^{-1}|$]\label[corollary]{cor:F-inverse-radius}
Taking $r=1.1$, $R=0.975$, and $M=0.98$, the inverse $F^{-1}$ is holomorphic on $0.975\D$ and
\[
|F^{-1}(\zeta)|<0.98\qquad (|\zeta|<0.975).
\]
\end{corollary}

\begin{proof}
Let
\[
\Delta(w):=F(w)-w=\sum_{n\ge1}d_nw^n.
\]
Using \cref{cor:coefficient-formula-combined}, the Taylor coefficients of \(H\) through degree \(144\) are determined using
\texttt{Arb}, and composed with the Taylor expansion of $\sin$. 
A composition of the series with the truncated Taylor expansion of $\sin$ through degree \(144\) in \texttt{Arb} gives the finite head estimate
\[
\sum_{n\le 144}|d_n|0.98^n<0.0045141060.
\]
By \cref{lemma:sine-coordinate-bound} and Cauchy's estimate, taking the Cauchy radius to $1.1$ from below gives
\[
\sum_{n>144}|d_n|0.98^n
\le
15.3\frac{(0.98/1.1)^{145}}{1-0.98/1.1}
< 7.4601\cdot10^{-6}.
\]
Therefore
\[
\sup_{|w|=0.98}|F(w)-w|
\le
\sum_{n\ge1}|d_n|0.98^n
<0.0045141060+7.4601\cdot10^{-6}
<0.005.
\]
Since $M-R=0.98-0.975=0.005$, \cref{lemma:rouche-inverse-disk} applies with $r=1.1$, $M=0.98$, and $R=0.975$.
\end{proof}

\begin{corollary}[Radius and supremum of $|H^{-1}|$]\label[corollary]{cor:H-inverse-radius}
Undoing the change of variables, we get
\[
H^{-1}=\sin\circ F^{-1}.
\]
Thus $H^{-1}$ is holomorphic on $0.975\D$ and
\[
|H^{-1}(\zeta)|<\sinh(0.98)<1.145
\qquad (|\zeta|<0.975).
\]
\end{corollary}

\begin{proof}
Since $F=H\circ\sin$, the identity $F^{-1}=\arcsin\circ H^{-1}$ is equivalent to
$H^{-1}=\sin\circ F^{-1}$ on the inverse branch at the origin.  By \cref{cor:F-inverse-radius}, $|F^{-1}(\zeta)|<0.98$ for $|\zeta|<0.975$.  The elementary bound
\[
|\sin z|\le \sinh |z|
\]
for complex $z$ gives
\[
|H^{-1}(\zeta)|=|\sin(F^{-1}(\zeta))|<\sinh(0.98)<1.145.
\]
\end{proof}

Finally, we are ready to prove \cref{lemma:tail-bound}.

\begin{lemma*}[Majorant Remainder bound]
Let
\[
H^{-1}(\zeta)=\sum_{n\ge1}a_n\zeta^n,
\qquad
\gzero=\log(1+\sqrt2)+4.8\cdot10^{-6}.
\]
Then
\[
\sum_{n>220}|a_n|\gzero^n<2.5\cdot10^{-9}.
\]
\end{lemma*}

\begin{proof}
By \cref{cor:H-inverse-radius}, $H^{-1}$ is holomorphic on $0.975\D$ and satisfies $|H^{-1}(\zeta)|<1.145$ there.  Hence, for every $s<0.975$, Cauchy's estimate with radius tending to $0.975$ gives
\[
\sum_{n>220}|a_n|s^n
\le
1.145\frac{(s/0.975)^{221}}{1-s/0.975}.
\]
Taking $s=\gzero$ yields
\[
1.145\frac{(\gzero/0.975)^{221}}{1-\gzero/0.975}
<2.44\cdot10^{-9}
<2.5\cdot10^{-9},
\]
as claimed.
\end{proof}

%% file: paper/head_computation.tex
\subsection{Computing the Head}
\label{sec:numerical-computation-of-head}

We now describe the certified computation of the head of the inverse-majorant
series for the \(\Hermite\) scheme.  

\begin{lemma}[Majorant head bound]\label[lemma]{lemma:head-bound}
Let
\[
    M_{f,g}(s)=\sum_{n\ge1}|a_n|s^n
\]
be the inverse majorant for the \(\Hermite\) scheme.  For
\[
    \gzero=\log(1+\sqrt2)+4.8\cdot10^{-6},
\]
one has
\[
    \sum_{1\le n\le220}|a_n|\gzero^n
    <1-2.5\cdot10^{-9}.
\]
\end{lemma}

\begin{proof}
Using \cref{cor:coefficient-formula-combined} and \cref{lemma:lagrange-inverse-coefficients}, the Taylor coefficients $(a_n)$ of
\(H^{-1}\) through degree \(220\) are determined using
\texttt{Arb} with outward-rounded interval arithmetic.  

The computation gives
\[
    \sum_{n\le220}|a_n|\gzero^n
    <
    0.999999729.
\]
Since
\[
    0.999999729
    <
    1-2.5\cdot10^{-9},
\]
the desired head bound follows.
\end{proof}

\begin{remark}
We also independently verified the computations using floating-point and PyTorch computations in Python. However, the displayed inequality is certified by the outward-rounded Arb computation.
\end{remark}

\begin{corollary}[$\gamma_0$-admissibility]\label[corollary]{cor:gamma0-admissibility}
For the $\Hermite$ scheme,
\[
M_{f,g}(\gzero)<1.
\]
Consequently, the scheme is $\gzero$-admissible.
\end{corollary}

\begin{proof}
By \cref{lemma:head-bound},
\[
\sum_{1\le n\le220}|a_n|\gzero^n<1-2.5\cdot10^{-9}.
\]
By \cref{lemma:tail-bound},
\[
\sum_{n>220}|a_n|\gzero^n<2.5\cdot10^{-9}.
\]
Adding the two inequalities gives $M_{f,g}(\gzero)<1$.  The inverse branch exists on a neighborhood of $\gzero\overline{\D}$ by \cref{cor:H-inverse-radius}, since $\gzero<0.975$.  Thus the $\Hermite$ scheme satisfies \cref{def:gamma-admissible}.
\end{proof}

\begin{proof}[Proof of \cref{thm:no-mixing-beats-hyperplane}]
By \cref{cor:gamma0-admissibility}, the $\Hermite$ scheme is $\gzero$-admissible.  Applying \cref{lemma:local-inverse-expansion} gives
\[
K_G\le \frac{\pi}{2\gzero}.
\]
Finally, with $\rhostarval=\log(1+\sqrt2)$ and $\gzero=\rhostarval+4.8\cdot10^{-6}$,
\[
\frac{\pi}{2\rhostarval}-\frac{\pi}{2\gzero}
=
\frac{\pi(4.8\cdot10^{-6})}{2\rhostarval\gzero}
>9.7\cdot10^{-6}.
\]
Therefore
\[
K_G
\le
\frac{\pi}{2\gzero}
<
\frac{\pi}{2\log(1+\sqrt2)}-9.7\cdot10^{-6}.
\]
\end{proof}

%% file: paper/generalization.tex
\section{Mixed Schemes}
\label{sec:generalization}

Our main result builds on the analysis technique developed in the prior sections but uses a mixture of two rounding schemes. We first introduce mixed schemes and then analyze them. 

First, we define a mixed scheme and recall its properties from~\cite{braverman2011grothendieckconstantstrictlysmaller, naor2014krivine}. For completeness, we also provide a proof in~\cref{app:krivine-preprocessing}.

\begin{theorem}[Mixed Krivine Scheme]\label[theorem]{thm:mixed-krivine}
Let $L\in\mathbb N$. For each $\ell\in\{1,\dots,L\}$, let
\(
f_\ell,g_\ell:\mathbb R^{k_\ell}\to\{-1,1\}
\)
be odd measurable functions, and let $H_\ell$ be their arcsine-normalized correlation function.
Let $\lambda_1,\dots,\lambda_L\ge 0$ satisfy $\sum_{\ell=1}^L\lambda_\ell=1$, and define the mixed correlation function
\(
H_\lambda(t):=\sum_{\ell=1}^L\lambda_\ell H_\ell(t).
\)
Assume that $H_\lambda$ has a local inverse at the origin,
\(
H_\lambda^{-1}(\zeta)=\sum_{n\ge1}a_n\zeta^n,
\)
and that this inverse is holomorphic on a neighborhood of $\gamma\overline{\mathbb D}$ for some
$\gamma>0$. If
\(
M_\lambda(\gamma):=\sum_{n\ge1}|a_n|\gamma^n\le 1,
\)
then
\(
K_G\le \frac{\pi}{2\gamma}.
\)
\end{theorem}
Then, one can use the following theorem to certify a bound on $K_G$ subject to finding the right parameters. 
\begin{theorem}[Mixed Krivine Criteria]\label[theorem]{thm:finite-mixed-function}
Let $H_\lambda=\sum_{\ell=1}^L\lambda_\ell H_\ell$ be a mixed correlation function as defined in
\cref{thm:mixed-krivine}. Define
\[
F(w):=H_\lambda(\sin w).
\]
Fix numbers
\(
0<R<M<r_0,
N_0, N_1 \in\mathbb N,
\gamma<R.
\)
Assume the following four conditions hold:
\begin{enumerate}
    \item \textbf{Sine-coordinate analytic bound.}
    The function $F$ is holomorphic on $|w|<r_0$, and
    \(
    |F(w)-w|\le B
    \)
    for \((|w|\le r_0).\)
    \item \textbf{Finite perturbation bound.}
    Writing
    \(
    F(w)-w=\sum_{n\ge1}d_nw^n,
    \)
    one has
    \(
    \sum_{n\le N_0}|d_n|M^n
    +
    B\frac{(M/r_0)^{N_0+1}}{1-M/r_0}
    <
    M-R.
    \)

    \item \textbf{Finite inverse-majorant head.}
    Writing
    \(
    H_\lambda^{-1}(\zeta)=\sum_{n\ge1}a_n\zeta^n,
    \)
    one has
    \(
    \sum_{n\le N_1}|a_n|\gamma^n
    \le
    1-T.
    \)

    \item \textbf{Cauchy remainder domination.}
    The remainder margin satisfies
    \(
    \sinh(M)\frac{(\gamma/R)^{N_1+1}}{1-\gamma/R}
    \le T.
    \)
\end{enumerate}

Then
\(
K_G\le \frac{\pi}{2\gamma}.
\)
\end{theorem}

\begin{proof}
By the finite perturbation bound,
\[
\sup_{|w|=M}|F(w)-w|<M-R.
\]
Rouch\'e's theorem therefore implies that $F^{-1}$ is holomorphic on $R\mathbb D$ and satisfies
\[
|F^{-1}(\zeta)|<M
\qquad (|\zeta|<R).
\]
Since
\[
H_\lambda^{-1}(\zeta)=\sin(F^{-1}(\zeta)),
\]
we get
\[
|H_\lambda^{-1}(\zeta)|
<
\sinh(M)
\qquad (|\zeta|<R).
\]
Cauchy's estimate gives
\[
\sum_{n>N_1}|a_n|\gamma^n
\le
\sinh(M)\frac{(\gamma/R)^{N_1+1}}{1-\gamma/R}
\le T.
\]
Combining this with the finite inverse-majorant head gives
\[
M_\lambda(\gamma)
=
\sum_{n\ge1}|a_n|\gamma^n
\le
(1-T)+T
\le 1.
\]
The conclusion now follows from \cref{thm:mixed-krivine}.
\end{proof}

\cref{thm:finite-mixed-function} can be used in conjunction with \texttt{Arb} to explicitly prove $\gamma$-admissibility of mixed rounding schemes. 

%% file: paper/ai_guided_search.tex

\cref{thm:pure-h3} gives us a conceptual improvement over Krivine's bound, but uses a suboptimal rounding scheme. In order to effectively explore the space of rounding schemes, we need an evaluator that accepts a rounding scheme as input and automatically outputs the best \(K_G\) bound. We implement this using \cref{thm:finite-mixed-function}, by checking that all the conditions in the statement are satisfied. 

Equipped with this evaluator, we can now perform an AI-guided search over the vast space of rounding schemes. The search consists
of the AI agent optimizing a rounding scheme template, using the obtained $K_G$ bound as numerical feedback for iteration.  The AI-guided search discovers our best bound on $K_G$, which we state below.

Let \(\mathrm{He}_k\) denote the $k$-th orthonormal probabilist's Hermite polynomial.

\begin{theorem}[Explicit mixed degree-\(9\) improvement]
\label{thm:main-explicit}
Let
\[
\begin{aligned}
P_9(x)
&=
-0.06924464693156676\,\mathrm{He}_1(x) 
-0.08372969497289807\,\mathrm{He}_3(x) \\
&\hspace{1em}
-0.034508730002003336\,\mathrm{He}_5(x)  
-0.030311217605625884\,\mathrm{He}_7(x) \\
&\hspace{1em}
-0.010654216877362276\,\mathrm{He}_9(x),
\\[0.4em]
Q_9(x)
&=
-0.06471449673854089\,\mathrm{He}_1(x)
+0.06896712471559421\,\mathrm{He}_3(x) \\
&\hspace{1em}
+0.006235636148918555\,\mathrm{He}_5(x) 
+0.0015193974910099376\,\mathrm{He}_7(x) \\
&\hspace{1em}
-0.0026674434971127943\,\mathrm{He}_9(x),
\end{aligned}
\]
and let
\[
    p=0.2733602555336593.
\]
Define 
\[
    f_{P_9}(x_1,x_2)=\sgn(x_2-P_9(x_1)),
    \qquad
    g_{Q_9}(x_1,x_2)=\sgn(x_2-Q_9(x_1)).
\]
Consider the rounding scheme which uses the hyperplane pair
\((\sgn(x_2),\sgn(x_2))\) with probability \(1-p\), and the Hermite-threshold
pair \((f_{P_9},g_{Q_9})\) with probability \(p\).  Let
\[
    \gamma = \rho_* + 2.987\cdot 10^{-5},
    \qquad
    \rho_*=\log(1+\sqrt2).
\]
Then the mixed scheme is $\gamma$-admissible, and therefore
\[
    K_G
    \le
    \frac{\pi}{2\gamma}
    \le
    \frac{\pi}{2\log(1+\sqrt2)}
    -
    6.039\cdot 10^{-5}.
\]
\end{theorem}

\begin{proof}
The displayed scheme is a mixed Krivine scheme, as defined in
\cref{thm:mixed-krivine}.  Therefore, it is sufficient to check that the criteria as described in 
\cref{thm:finite-mixed-function} hold. We use the following values of the parameters required by \cref{thm:finite-mixed-function},
\[R = 0.975,\; M = 0.98,\; r_0 = 1.1, \qquad N_0 = 106,\; N_1 = 192, \qquad \gamma = \rho_* + 2.987\cdot 10^{-5}. \]
We now verify each of the criteria.
\begin{enumerate}
    \item \textbf{Sine-coordinate analytic bound}. Let $B = 15.3$. By~\cref{lemma:sine-coordinate-bound}, we have $|F(w) - w|  < B$ for $|w| < r_0$, thus this condition is satisfied.  

    \item \textbf{Finite perturbation bound}. We use \cref{cor:coefficient-formula-combined} to compute the Taylor coefficients of
\(H\) through degree \(N_0 = 106\) with
\texttt{Arb}, and compose with the Taylor expansion of $\sin$ to obtain $d_i$. 
This gives the finite head estimate
\[
\sum_{n\le N_0}|d_n|0.98^n<0.0034656.
\]
By \cref{lemma:sine-coordinate-bound} and Cauchy's estimate, taking the Cauchy radius to $r_0 = 1.1$ from below gives
\[
\sum_{n>N_0}|d_n|0.98^n
\le
15.3\frac{(0.98/1.1)^{N_0+1}}{1-0.98/1.1}
< 0.0006013 .
\]
Therefore
\[
\sup_{|w|=0.98}|F(w)-w|
\le
\sum_{n\ge1}|d_n|0.98^n
< 0.0034656 + 0.0006013 
< 0.005 = M - R,
\]
as desired. 

\item \textbf{Finite inverse-majorant head}. Let $T=4.9298 \cdot 10^{-8}$. A computation in \texttt{Arb} gives the head, \[\sum\limits_{n \le N_1} |a_n| \gamma^n \le 0.9999999507019  <  1 - T \]

\item \textbf{Cauchy remainder domination}. By Cauchy estimates, the remainder is at most 
\[
    \sinh(M)\frac{(\gamma/R)^{N_1+1}}{1-\gamma/R}
    \le 4.1422 \cdot  10^{-8} \le 4.9298 \cdot  10^{-8} = T.
\]
\end{enumerate}
Thus, by~\cref{thm:finite-mixed-function}, we have
\[
    K_G
    \le
    \frac{\pi}{2\gamma}
    \le
    \frac{\pi}{2\log(1+\sqrt2)}
    -
    6.039\cdot 10^{-5}.
\]
\end{proof}

\input{figures/best_deg9_scheme}

\subsection{Search Procedure}
\label{sec:explicit-improvement}

\label{sec:explicit-search}

We now describe the search. The search was implemented using an AI coding agent (Claude Code, Opus 4.7 \cite{anthropic:2026:opus4.7}). 

\paragraph{Search template.} For a fixed odd degree \(d\), we search over pairs of odd Hermite thresholds
\[
    P(x)=\sum_{\substack{1\le k\le d\\ k\ \mathrm{odd}}} a_k\,\mathrm{He}_k(x),
    \qquad
    Q(x)=\sum_{\substack{1\le k\le d\\ k\ \mathrm{odd}}} b_k\,\mathrm{He}_k(x),
\]
which define planar partitions
\[
    f_P(x_1,x_2)=\sgn(x_2-P(x_1)),
    \qquad
    g_Q(x_1,x_2)=\sgn(x_2-Q(x_1)).
\]
The search also allows mixtures with other rounding schemes. Thus, a candidate is specified by the
coefficient vectors \((a_k)\), \((b_k)\), together with a mixing weight vector $p$.

\paragraph{Landscape.}  The search variables for this optimization are the coefficient vectors $(a_k), (b_k)$, and the simplex weights $p \in \simplex^{k-1}$. The
  end-to-end map from these parameters to the bound on $K_G$ consists of three
  sources of non-smoothness: 
  (1) the partitions $f_P, g_Q$ are sign
  functions of polynomial inputs and so depend on the coefficients
  discontinuously;
  (2) the inverse-majorant
  $M(\gamma) = \sum_n |a_n|\,\gamma^n$ takes coefficient-wise
  absolute values; and 
  (3) $K_G$ itself is defined implicitly as a function of the
  unique root of $M(\gamma) = 1$. 
  End-to-end backpropagation through this mapping is highly non-trivial. 

  \paragraph{Search strategy.} The AI agent proposes a strategy that decouples the search into two stages with clean gradient structures. Stage~1 fits a single
  partition pair $(P, Q)$ jointly with its hyperplane mixing weight by using the
  Adam optimizer on the surrogate
  $(a_k, b_k, p) \mapsto M(\gamma)$,
  which is piecewise smooth in the
  coefficients. Stage~2 freezes
  the resulting trained candidates and maximises $\gamma(p)$
  over the unit simplex via SLSQP (Sequential Least Squares Programming \cite{kraft1988software}), which is purpose-built for
  low-dimensional nonlinearly constrained problems with implicit
  objectives. We describe the stages in detail below. 



\paragraph{Stage 1: training individual candidates.}
For each degree \(d\), we initialize the Hermite coefficients
\((a_k),(b_k)\) randomly and optimize a floating-point surrogate for the
inverse-majorant objective.  The surrogate is obtained by computing a truncated
Taylor series for the correlation function, inverting this series numerically,
and evaluating the truncated majorant \(M(\gamma)\) at a target value
\(\gamma>\rho_*\).  In practice we minimize this surrogate loss using Adam in
\texttt{float64}.  The hyperplane weight is trained simultaneously through a
logit parametrization, so that the search can decide how much mass to place on
the non-hyperplane partition.

For each degree we run many independent random seeds and retain the best
trained candidates according to their estimated value of \(\gamma\).  These candidates form a finite dictionary for the second stage.

\paragraph{Stage 2: optimizing the mixture weights.}
After Stage 1, the candidates are fixed.  The remaining problem is finite
dimensional: choose a probability vector on the dictionary so as to maximize
the estimated inverse-majorant root \(\gamma\).  Since the mixed correlation
function is linear in the mixture weights, each evaluation of a proposed mixture
only requires forming the mixed Taylor series, inverting it, and evaluating the
majorant.  We solve this simplex-constrained optimization problem using SLSQP
from many Dirichlet-random initializations, keeping the best local optimum.

The output of Stage 2 is a small explicit mixture.  For the final degree-\(9\)
certificate, this reduces to a two-candidate mixture: the hyperplane scheme and one
degree-\(9\) Hermite-threshold pair.  The lower-degree rows in
\cref{sec:explicit-degree-comparison} are obtained by the same procedure with
the maximum odd degree fixed to \(d=3,5,7\), respectively.





Here, we note that the search produces only candidates. Each are then formally certified in \texttt{Arb} using \cref{thm:finite-mixed-function}. 

\paragraph{Reproducibility}
The coefficient vectors, mixture weights, Arb outputs, and scripts are stored in the project repository at \url{https://github.com/trishullab/grothendieck-rounding-repo}.

\subsection{Discussion}

\paragraph{How do higher degrees help?}
\label{sec:explicit-degree-comparison}

The same procedure was run for maximum odd Hermite degree
\(d\in\{3,5,7,9\}\).  Each row below is an independent search with degree
\(d\) imposed from the start; the lower-degree rows are not truncations of the
degree-\(9\) solution.  The coefficient vectors are listed in
\cref{app:explicit-mixture-coefficients}.
\begin{center}
\renewcommand{\arraystretch}{1.2}
\small
\setlength{\tabcolsep}{8pt}
\begin{tabular}{cc}
\hline
\(d\)
&
\(K_G\) improvement
\\
\hline
\(3\)
&
\(2.837\cdot10^{-5}\)
\\
\(5\)
&
\(5.571\cdot10^{-5}\)
\\
\(7\)
&
\(5.955\cdot10^{-5}\)
\\
\(9\)
&
\(6.039\cdot10^{-5}\)
\\
\hline
\end{tabular}
\end{center}
The table shows a clear diminishing return from having additional degrees. 
Passing from degree
\(3\) to degree \(5\) gives the largest gain, with subsequent degrees only offering modest additional
increase.  This suggests that most of the benefit in this search family is
already captured by low Hermite degrees.

\paragraph{Why does mixing improve the Grothendieck bound?}
\label{sec:discussion-mixing}

One of the findings of our work is that a randomized mixture of rounding schemes can
sometimes give a \emph{strictly} better Grothendieck bound than the component schemes.
Intuitively, different schemes can perform well on different difficult configurations of
vectors, so the right search space is not just the set of individual partitions, but also the
simplex of mixtures of promising partitions.

The mechanism by which this happens is visible in the inverse series. For a two-scheme mixture, the correlation
coefficients are linear in the mixing parameter,
$b_{2m+1}(p)=(1-p)b^{(0)}_{2m+1}+p b^{(1)}_{2m+1}$. The inverse coefficients in
$H_p^{-1}(\zeta)=\sum_{m\ge 0}a_{2m+1}(p)\zeta^{2m+1}$, however, are not linear in $p$.
Comparing coefficients in $H_p(H_p^{-1}(\zeta))=\zeta$ determines them recursively, so
$a_{2m+1}(p)$ is a rational function of $b_1(p),b_3(p),\ldots,b_{2m+1}(p)$, with denominators
involving powers of $b_1(p)$. Thus, wherever $b_1(p)=H_p'(0)\neq 0$, the functions
$a_{2m+1}(p)$ are real analytic in $p$. Since the bound is controlled by
\[
    M_p(\gamma)=\sum_{m\ge 0}|a_{2m+1}(p)|\gamma^{2m+1}\le 1,
\]
by continuity a sign change in an inverse coefficient can be converted into an actual cancellation, removing
one term from the absolute-value series.

A simple example is obtained by mixing the hyperplane scheme with the perturbation scheme
\[
    f(x,y)=\operatorname{sgn}(y-0.05\,\mathrm{He}_3(x)),\qquad
    g(x,y)=\operatorname{sgn}(y+0.05\,\mathrm{He}_3(x)),
\]
With $p$ denoting the perturbation weight, the numerical
one-parameter optimum occurs at the root $a_7(p)=0$, namely $p=0.798782\ldots$. The resulting
bound is
\[
    \frac{\pi}{2\gamma}=1.78219809,
\]
which is smaller than both the hyperplane scheme value $1.78221398$ and the pure perturbation scheme value $1.78225030$. This reflects the broader empirical pattern in our mixed searches: the
best improvements occur when the two endpoint schemes have opposite signs for $a_7$, so that
continuity forces a mixture for which the $a_7$-term vanishes. It remains an interesting
question whether a better mixture can cancel an earlier coefficient, such as $a_5$, while
keeping the lower-order $a_1$ and $a_3$ terms essentially unchanged.

%% file: figures/best_deg9_scheme.tex
\begin{figure}
    \centering
    \includegraphics[width=0.9\linewidth]{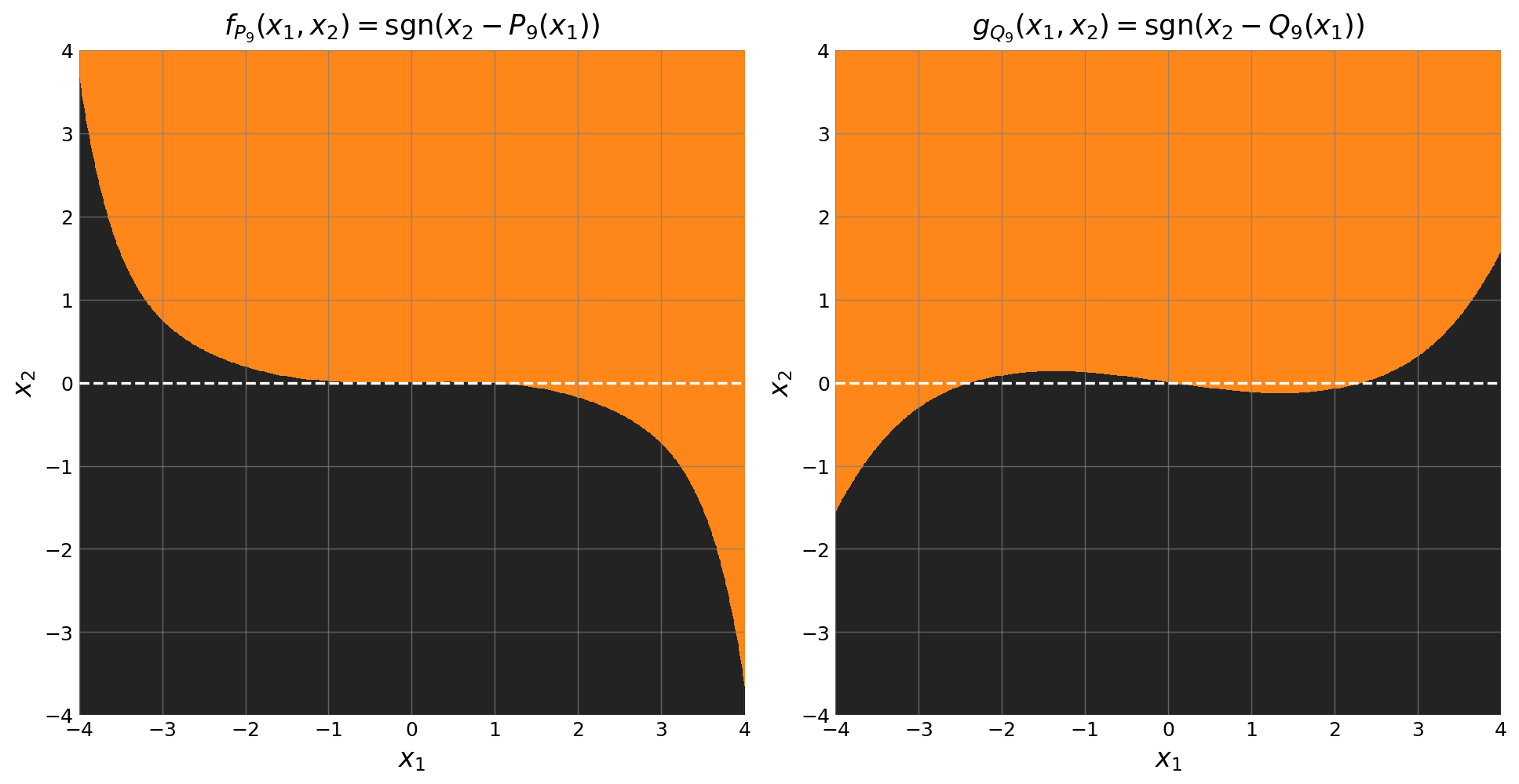}
    \caption{The degree 9 partitions. The resulting scheme is mixed with the hyperplane scheme to obtain our bound.}
    \label{fig:best-deg9-scheme}
\end{figure}

%% file: paper/tiger_partitions.tex
\section{K\"onig's Bilinear Form and Tiger Partitions}
\label{sec:tiger-partitions}
In this section, we define K\"onig's bilinear form, which is closely associated with the Grothendieck constant. 

\begin{definition}[K\"onig's bilinear form]\label{def:konig}
For odd partitions $f,g:\mathbb R^k\to\{-1,1\}$, define
\begin{equation}
  \widetilde{\mcal{B}}_K(f,g)
  := \int_{\mathbb R^k}\int_{\mathbb R^k}
       f(x)g(y)\exp\left(-\frac{\|x\|^2+\|y\|^2}{2}\right)
       \sin\langle x,y\rangle\,dx\,dy .
  \label{eq:konig-bilinear-form}
\end{equation}
We use the normalized value
\begin{equation}
  \mcal{B}_K(f,g):=\frac{\widetilde{\mcal{B}}_K(f,g)}{(\sqrt{2}\pi)^k},
\end{equation}
so that the hyperplane partition $f_0(x)=g_0(x)=\operatorname{sign}(x_0)$ (with $x_0$ the first coordinate of $x$) has
\begin{equation}
  \mcal{B}_K(f_0,g_0)=\frac{2}{\pi}\log(1+\sqrt 2).
  \label{eq:krivine-cK}
\end{equation}
\end{definition}

For an odd partition \(f:\R^k\to\{-1,1\}\), define its
standard-Gaussian rescaling by
\(
    f^\sharp(x):=f\left(x/\sqrt{2}\right).
\)
A change of variables gives
\(
    \frac{H_{f^\sharp,g^\sharp}(i)}{i}
    =
    \frac{\pi}{2}\mcal B_K(f,g).
\)
The study of $\mcal{B}_K$ in connection to the Grothendieck inequality is motivated by the following result from~\cite{braverman2011grothendieckconstantstrictlysmaller}, that we paraphrase:

\begin{theorem}
\label{thm:BK_grothendieck}
    Let $(f,g)$ be odd partitions, and assume $\mcal{B}_K (f, g)$ {\scriptsize $> \tfrac{2}{\pi}\log(1 + \sqrt{2})$}. For $0 < p < 1$, consider the mixed scheme given by choosing the $(f^\sharp, g^\sharp)$-scheme with probability $p$ and the hyperplane scheme with probability $1-p$. Then for all sufficiently small $p$, there exists some $\gamma_p > \log(1 + \sqrt{2})$ such that the mixed scheme is $\gamma_p$-admissible. In particular, this implies \(K_G \le \frac{\pi}{2\gamma_p} < \tfrac{\pi}{2\log(1 + \sqrt{2})}.\) 
\end{theorem}

Therefore, it is natural to consider whether maximizers of K\"onig's bilinear form also translate to good improvements on upper bounds on $K_G$. Below, we show numerical evidence against this. 

Following~\cite{braverman2011grothendieckconstantstrictlysmaller}, define the best-response operator
\begin{equation}
  \sigma(f)(y)
  := \operatorname{sign}\left(\int_{\mathbb R^k}
        f(x)\exp\left(-\frac{\|x\|^2}{2}\right)
        \sin\langle x,y\rangle\,dx\right),
  \label{eq:sigma-operator}
\end{equation}
with an arbitrary convention on the null set where the argument vanishes.  Since $\mcal{B}_K$ is linear in each argument, $\sigma(f)$ is the pointwise best response to $f$, so any maximizing pair should satisfy $\sigma(f_{\max})=g_{\max}$ and $\sigma(g_{\max})=f_{\max}$. \cite{braverman2011grothendieckconstantstrictlysmaller} numerically found examples of these maximizing pairs by repeatedly applying the $\sigma$ operator to a series of initial partitions, and found that they all converged (with the exception of the hyperplane fixed point) to the two-cycle partitions $f_\infty$ and $g_\infty$ with an interesting striped pattern (see \cref{fig:tiger-partitions}), which they named \textit{tiger partitions}.

\input{figures/figure_tiger_partitions}

They posed an open question as to whether the partitions $(f_\infty,g_\infty)$ maximize $B_K$, and whether the corresponding $(f_\infty^\sharp,g_\infty^\sharp)$ are the optimal partitions for the Grothendieck constant. We show numerical evidence the latter is false for $k=2$.

To replicate the \cite{braverman2011grothendieckconstantstrictlysmaller} generation procedure, we started from $1000$ random odd partitions from a variety of odd partition families, and repeatedly applied $\sigma$ on a finite $[-8,8]^2$ grid, stopping when the iterates converged to a fixed point or entered a stable oscillation/two-cycle.  For each terminal pair $(f,g)$ we computed both the normalized K\"onig value $\mcal{B}_K(f,g)$ and the upper bound on $K_G$ obtained from the rescaled scheme ($\tfrac{\pi}{2\gamma}$). 


  \begin{table}[h]
  \centering
  \begin{tabular}{lrrrr}
  \hline
  Quantity & Min & Max & Mean & Median \\
  \hline
  $\mcal{B}_K(f,g)$ & $0.561910$ & $0.561923$ & $0.561917$ & $0.561917$ \\
  $K_G\;\;(\le \frac{\pi}{2\gamma})$ & $1.7921$ & $1.7938$ & $1.7931$ & $1.7930$ \\
  \hline
  \end{tabular}
  \caption{Aggregate statistics over random odd initial partitions after
  repeated $\sigma$-iteration}
  \label{tab:tiger-summary-statistics}
  \end{table}

The table shows that repeated $\sigma$-iteration reliably produces tiger partitions with high
K\"onig value: the resulting values concentrate around
\(\mcal{B}_K(f,g)\approx 0.5619\), above the hyperplane baseline
\(
    \frac{2}{\pi}\log(1+\sqrt 2)\approx 0.5611.
\)
However, the corresponding rescaled partitions perform poorly as Grothendieck rounding schemes. The 
upper bounds on \(K_G\) are typically around \(1.7930\), worse than Krivine's bound
\(
    \frac{\pi}{2\log(1+\sqrt 2)}\approx 1.7822.
\)
Thus, across all trials, the $\sigma$-generated tiger partitions improve K\"onig's bilinear form 
but the associated pair never improve the Grothendieck bound.

%% file: figures/figure_tiger_partitions.tex
\begin{figure}
    \centering
    \includegraphics[width=0.85\linewidth]{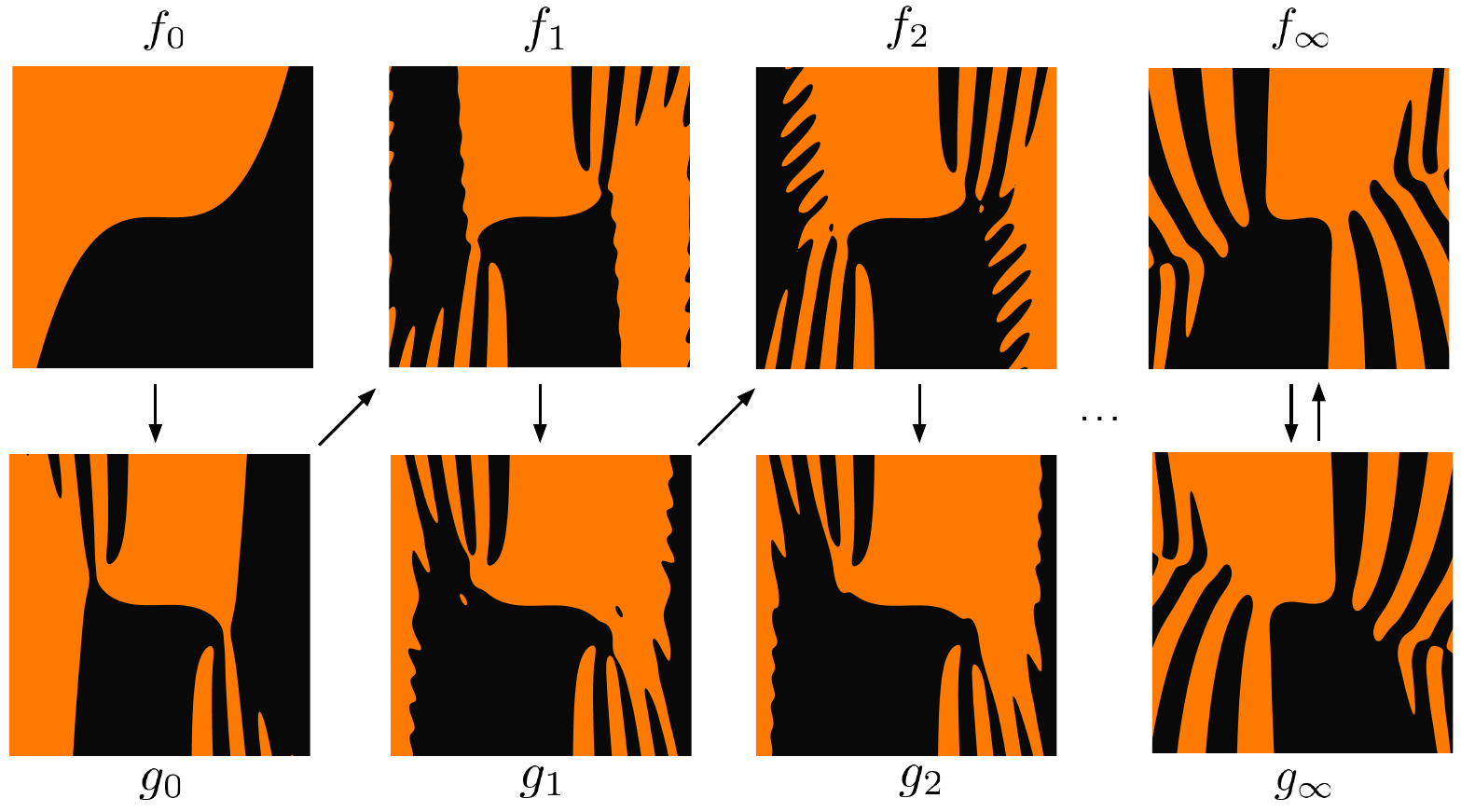}
    \caption{Iteratively applying the $\sigma$ operator to a random initial partition $f_0$. Converges to the two-cycle tiger partitions $(f_\infty, g_\infty)$, where $\sigma(f_\infty)=g_\infty$ and $\sigma(g_\infty)=f_\infty$.}
    \label{fig:tiger-partitions}
\end{figure}

%% file: paper/koenig.tex
\subsection{Better Space Partitions for K\"onig's Bilinear Form}
\label{sec:konig-not-sufficient}


\cref{sec:tiger-partitions} shows that a large K\"onig value need not automatically produce a better Krivine rounding scheme.  Nevertheless, K\"onig's form remains an interesting object because it isolates a clean problem for non-hyperplane partitions and captures the first step from \cite{Kon00,braverman2011grothendieckconstantstrictlysmaller}.
In this section, we show that K\"onig's bilinear form admits an explicit family of partitions whose normalized value in the limit of dimension significantly outperforms the hyperplane benchmark of 
\(
 \frac{2}{\pi}\log(1+\sqrt2) = 0.56109 \ldots
\)
The family keeps a two-dimensional active half-space, but rotates it by an affine function of the squared radius in an auxiliary $d$-dimensional Gaussian block:

\begin{theorem}\label{th:konignain}
    For every dimension $d$, there exist odd partitions $f_d,g_d: \mathbb{R}^{d+2} \rightarrow \{1,-1\}$, such that $\lim_{d \rightarrow \infty} \mcal{B}_K(f_d,g_d) \geq 0.59357\cdots$
\end{theorem}



The construction is inspired by Krivine's 2023 paper \cite{krivine2023notegrothendiecksconstant}, but goes far beyond it.  Krivine considers a three-dimensional rotating half-space: two active Gaussian coordinates define the sign cut, while one auxiliary coordinate, through $H_2(x)=x^2-1$, controls the rotation angle.  \cite{krivine2023notegrothendiecksconstant}  then computes the corresponding K\"onig value in this three-dimensional setting. Our contribution is to introduce a general auxiliary dimension and analyze the asymptotic regime.  Here we keep Krivine's rotating-half-space mechanism but replace the single auxiliary coordinate by a radial observable in $\R^d$, derive an exact finite-$d$ Fourier formula, and pass to a cubic Fourier-multiplier limit as $d\to\infty$.

\subsection{The Higher-Dimensional Affine Radial Family}

Recall the definition of $\mcal{B}_K(f,g)$ for odd measurable partitions $f,g:\mathbb{R}^k \rightarrow \{-1,1\}$ from Definition \ref{def:konig}. 

\ignore{
For odd measurable partitions $f,g:\R^k\to\{-1,1\}$, write
\[
 H_{f,g}(t):=\frac{\pi}{2}\,\E[f(X)g(Y)],
 \qquad |t|<1,
\]
where $X,Y$ are standard Gaussian vectors with coordinatewise correlation $t$.  We use the normalized K\"onig value
\[
 \mathcal B_{\mathrm K}(f,g):=\frac{2}{\pi}\,\frac{H_{f,g}(i)}{i},
\]
where $H_{f,g}$ is analytically continued to $t=i$.  The active two-dimensional calculation leaves the fixed circle profile
\[
 \mathfrak a(\theta):=\arsinh(\cos\theta),
 \qquad \theta\in \R/2\pi\Z.
\]
Thus the hyperplane partition has normalized value
\[
 \mathcal B_{\mathrm K}(\mathrm{hp},\mathrm{hp})
 =\frac{2}{\pi}\mathfrak a(0)
 =\frac{2}{\pi}\log(1+\sqrt2).
\]

We now define the affine radial family.}

\begin{definition}[Affine radial family]
    
Fix $d\in\mathbb N$ and parameters $\alpha,\beta\in\R$.  Let
\[
 Q_d(u):=4\norm{u}_2^2-d,
 \qquad u\in\R^d,
\]
\[
 \psi_{d,\alpha,\beta}(u):=\beta+\alpha Q_d(u),
\]
and
\[
 \phi_{d,\alpha,\beta}(u,v)
 :=\psi_{d,\alpha,\beta}(u)+\psi_{d,\alpha,\beta}(v)
 =2\beta+\alpha\bigl(Q_d(u)+Q_d(v)\bigr).
\]

The affine radial family is $(f_{d,\alpha,\beta}, g_{d,\alpha,\beta})$ defined on $\R^{d+2}=\R^d\times\R^2$ parametrized as follows:

\[
 f_{d,\alpha,\beta}(u;x_1,x_2)
 :=
 \operatorname{sign}\bigl(x_1\cos\psi_{d,\alpha,\beta}(u)+x_2\sin\psi_{d,\alpha,\beta}(u)\bigr),
\]
\[
 g_{d,\alpha,\beta}(v;y_1,y_2)
 :=
 \operatorname{sign}\bigl(y_1\cos\psi_{d,\alpha,\beta}(v)-y_2\sin\psi_{d,\alpha,\beta}(v)\bigr).
\]
\end{definition}

Here $\beta$ is a constant rotation on each side, while $\alpha$ controls the radial dependence.  The factor $4$ in $Q_d$ is the one produced by the $t=i$ Gaussian--oscillatory kernel: writing auxiliary coordinates as $X_j=2u_j$ turns $H_2(X_j)=X_j^2-1$ into $4u_j^2-1$.

Conditioned on $u,v$, the active-plane correlation is $t\cos\phi_{d,\alpha,\beta}(u,v)$.  Evaluating at $t=i$ and integrating out the active plane gives the following reduced form.

\begin{definition}[Reduced affine radial K\"onig bilinear form]
Let
\[
 c_d:=\frac{2^{1+d/2}}{\pi^{d+1}},
\] and let
\[
 \mathfrak a(\theta)=\arsinh(\cos\theta),
 \qquad \theta\in\T.
\]
Define
\begin{equation}\label{eq:Kd-def}
\begin{aligned}
 \K_d(\alpha,\beta)
 &:=
 \mathcal B_{\mathrm K}(f_{d,\alpha,\beta},g_{d,\alpha,\beta}) \\
 &=
 c_d
 \int_{\R^d}\int_{\R^d}
 e^{-(\norm{u}_2^2+\norm{v}_2^2)}
 \cos\bigl(2\ip{u}{v}\bigr)
 \mathfrak a\bigl(\phi_{d,\alpha,\beta}(u,v)\bigr)
 \,du\,dv.
\end{aligned}
\end{equation}
\end{definition}

When $\alpha=\beta=0$ there is no rotation, and one recovers the hyperplane benchmark
\[
 \K_d(0,0)=\frac{2}{\pi}\log(1+\sqrt2).
\]


The limiting object is a Fourier multiplier on the circle. 

\begin{definition}[Cubic Fourier multiplier/periodic Airy propagator]\label{def:cubic-airy}
Let $q(\theta)=\sum_{m\in\mathbb Z}\widehat q(m)e^{im\theta}$ be a $2\pi$-periodic analytic function, with Fourier coefficients
\[
 \widehat q(m):=\frac{1}{2\pi}\int_0^{2\pi} q(\theta)e^{-im\theta}\,d\theta.
\]
For $\tau\in\R$ define
\[
 \bigl(U(\tau)q\bigr)(\theta)
 :=
 \sum_{m\in\mathbb Z}e^{i\tau m^3}\widehat q(m)e^{im\theta}.
\]
Equivalently, since $-\partial_\theta^3 e^{im\theta}=im^3e^{im\theta}$, we have
\[
 U(\tau)=e^{-\tau\partial_\theta^3},
\]
and defining $u(\tau,\theta):=(U(\tau)q)(\theta)$ solves
\[
 \partial_\tau u+\partial_\theta^3 u=0,
 \qquad
 u(0,\theta)=q(\theta).
\]
Note that this is the linear Airy equation on $\T$ \cite[Ch.~2]{Tao2006NDE}.

\end{definition}

We can now state the main theorem for radial families.

\begin{theorem}\label{thm:main}
Fix $0<\rho<\rho_*$, where $\rho_*:=\arsinh(1)=\log(1+\sqrt2)$, and fix $\kappa_0>0$. Then there exist constants $C=C(\rho,\kappa_0)$ and $d_0=d_0(\rho,\kappa_0)$ such that for all $d\ge d_0$, all $|\kappa|\le\kappa_0$, and all $\beta\in\R$,
\begin{equation}\label{eq:main-estimate}
\left|
 \K_d\bigl(\kappa d^{-1/3},\beta\bigr)
 -
 \frac{2}{\pi}\bigl(e^{-\tau_\kappa\partial_\theta^3}\mathfrak a\bigr)(2\beta)
\right|
\le C d^{-1/3},
\end{equation}
where
\[
 \tau_\kappa:=\frac{16}{3}\kappa^3.
\]
Equivalently,
\[
 \K_d\bigl(\kappa d^{-1/3},\beta\bigr)
 \longrightarrow
 \frac{2}{\pi}\bigl(e^{-\tau_\kappa\partial_\theta^3}\mathfrak a\bigr)(2\beta)
\]
uniformly for $(\kappa,\beta)\in[-\kappa_0,\kappa_0]\times\R$.
\end{theorem}

\begin{remark}[Why the scaling is $d^{-1/3}$]
The proof diagonalizes $\K_d$ in the Fourier basis and produces a scalar multiplier $\Lambda(m\alpha)^d$ on the $m$th Fourier mode. The key expansion is
\[
 \log\Lambda(z)=\frac{16i}{3}z^3-16z^4+O(z^5).
\]
The first nontrivial oscillatory term is cubic. Therefore the only scaling for which $d\log\Lambda(m\alpha)$ remains nondegenerate is $d\alpha^3\asymp 1$, i.e.
\[
 \alpha\asymp d^{-1/3}.
\]
The cubic phase produces the limiting multiplier $e^{i\tau m^3}$, while the quartic term yields the $O(d^{-1/3})$ error.
\end{remark}


The finite-dimensional problem is already completely explicit.

\begin{theorem}[Exact reduction]\label{prop:exact}
Define
\[
 \Lambda(z):=\frac{e^{-2iz}}{\sqrt{1-4iz-8z^2}},
 \qquad z\in\R,
\]
with the principal branch of the square root. Then
\begin{equation}\label{eq:exact-reduction}
 \K_d(\alpha,\beta)
 =
 \frac{2}{\pi}
 \sum_{m\in\mathbb Z}
 \widehat{\mathfrak a}(m)e^{2im\beta}\Lambda(m\alpha)^d.
\end{equation}
\end{theorem}

\begin{proof}
Because $\mathfrak a$ is real analytic on a strip around the real axis, its Fourier series converges absolutely:
\[
 \mathfrak a(\theta)=\sum_{m\in\mathbb Z}\widehat{\mathfrak a}(m)e^{im\theta}.
\]
Substituting this into \eqref{eq:Kd-def} and interchanging the sum and the integral yields
\[
 \K_d(\alpha,\beta)
 =
 c_d
 \sum_{m\in\mathbb Z}\widehat{\mathfrak a}(m)e^{2im\beta}
 I_{m,d}(\alpha),
\]
where
\[
 I_{m,d}(\alpha)
 :=
 \int_{\R^d}\int_{\R^d}
 e^{-(\norm{u}^2+\norm{v}^2)}
 \cos\bigl(2\ip{u}{v}\bigr)
 e^{im\alpha(Q_d(u)+Q_d(v))}
 \,du\,dv.
\]
Since $Q_d(v)$ is even in $v$, the change of variables $v\mapsto -v$ shows that the two exponentials arising from
\[
 \cos\bigl(2\ip{u}{v}\bigr)=\frac{e^{2i\ip{u}{v}}+e^{-2i\ip{u}{v}}}{2}
\]
contribute equally. Hence
\[
 I_{m,d}(\alpha)
 =
 \int_{\R^d}\int_{\R^d}
 e^{-(\norm{u}^2+\norm{v}^2)}
 e^{2i\ip{u}{v}}
 e^{im\alpha(Q_d(u)+Q_d(v))}
 \,du\,dv.
\]
Now the integrand factorizes coordinate by coordinate, so
\[
 I_{m,d}(\alpha)
 =
 \prod_{j=1}^d
 \int_{\R^2}
 e^{-(x^2+y^2)}
 e^{im\alpha((4x^2-1)+(4y^2-1))}
 e^{2ixy}
 \,dx\,dy.
\]
The one-coordinate Gaussian integral is
\[
 \int_{\R^2}
 e^{-(1-4im\alpha)x^2-(1-4im\alpha)y^2+2ixy}
 \,dx\,dy
 =
 \frac{\pi}{\sqrt{(1-4im\alpha)^2+1}}.
\]
After including the factor $e^{-2im\alpha}$ coming from the two constant terms in $(4x^2-1)+(4y^2-1)$, the one-coordinate contribution is
\[
 \frac{\pi e^{-2im\alpha}}{\sqrt{(1-4im\alpha)^2+1}}
 =
 \frac{\pi}{\sqrt2}\,\Lambda(m\alpha).
\]
Therefore
\[
 I_{m,d}(\alpha)=\left(\frac{\pi}{\sqrt2}\right)^d\Lambda(m\alpha)^d.
\]
Since
\[
 c_d\left(\frac{\pi}{\sqrt2}\right)^d=\frac{2}{\pi},
\]
we obtain \eqref{eq:exact-reduction}. \qedhere
\end{proof}

First, we prove two preparatory lemmas. 

\begin{lemma}[Fourier coefficient decay]\label{lem:fourier-decay}
Let
\[
 \rho_*:=\arsinh(1)=\log(1+\sqrt2).
\]
Then for every $0<\rho<\rho_*$ there exists $C_\rho<\infty$ such that
\[
 \abs{\widehat{\mathfrak a}(m)}\le C_\rho e^{-\rho\abs m}
 \qquad (m\in\mathbb Z).
\]
In particular,
\[
 \sum_{m\in\mathbb Z}\abs{\widehat{\mathfrak a}(m)}\abs m^k<\infty
 \qquad\text{for every }k\ge 0.
\]
\end{lemma}

\begin{proof}
The singularities of $\arsinh(z)$ are at $z=\pm i$, so the singularities of
\[
 \mathfrak a(z)=\arsinh(\cos z)
\]
occur where $\cos z=\pm i$. The nearest such points to the real axis lie at imaginary distance $\rho_*$. Hence $\mathfrak a$ extends holomorphically to every strip
\[
 \abs{\Im z}<\rho
 \qquad (0<\rho<\rho_*).
\]
The standard Cauchy estimate for Fourier coefficients on analytic strips gives
\[
 \abs{\widehat{\mathfrak a}(m)}\le C_\rho e^{-\rho\abs m}.
\]
The summability of polynomially weighted coefficients follows immediately. \qedhere
\end{proof}

\begin{lemma}[Local expansion of the multiplier]\label{lem:lambda-expansion}
For $\abs z\le 1/16$,
\[
 \log\Lambda(z)=\frac{16i}{3}z^3-16z^4+r(z),
\]
with the explicit bound
\[
 \abs{r(z)}\le 500\abs{z}^5.
\]
\end{lemma}

\begin{proof}
Write
\[
 \Lambda(z)=e^{-2iz}(1+w(z))^{-1/2},
 \qquad
 w(z):=-4iz-8z^2.
\]
If $\abs z\le 1/16$, then
\[
 \abs{w(z)}\le 4\abs z+8\abs z^2\le \frac14+\frac1{32}=\frac9{32}<1.
\]
Hence the principal logarithm is analytic there and
\[
 \log\Lambda(z)=-2iz-\frac12\log(1+w(z)).
\]
Use the fourth-order logarithm expansion
\[
 \log(1+w)=w-\frac{w^2}{2}+\frac{w^3}{3}-\frac{w^4}{4}+R_5(w),
 \qquad
 \abs{R_5(w)}\le \frac{\abs w^5}{5(1-\abs w)}.
\]
A direct algebraic expansion gives
\[
\begin{aligned}
 -2iz-\frac12\left(w-\frac{w^2}{2}+\frac{w^3}{3}-\frac{w^4}{4}\right)
 &=
 \frac{16i}{3}z^3-16z^4-128iz^5 \\
 &\quad -\frac{2048}{3}z^6+1024iz^7+512z^8.
\end{aligned}
\]
Thus for $\abs z\le 1/16$,
\[
 \left|-128iz^5-\frac{2048}{3}z^6+1024iz^7+512z^8\right|
 \le 175\abs z^5.
\]
On the other hand, since $\abs{w(z)}\le \frac92\abs z$,
\[
 \frac12\abs{R_5(w)}
 \le
 \frac12\cdot\frac{\abs w^5}{5(1-\abs w)}
 \le 261\abs z^5.
\]
Combining the two bounds yields
\[
 \abs{r(z)}\le (175+261)\abs z^5\le 500\abs z^5.
\]
\qedhere
\end{proof}


We are now ready to prove Theorem~\ref{thm:main}. 
\begin{proof}[Proof of Theorem~\ref{thm:main}]
Start from Theorem~\ref{prop:exact}:
\[
 \K_d\bigl(\kappa d^{-1/3},\beta\bigr)
 =
 \frac{2}{\pi}
 \sum_{m\in\mathbb Z}
 \widehat{\mathfrak a}(m)e^{2im\beta}
 \Lambda\bigl(m\kappa d^{-1/3}\bigr)^d.
\]
Define the target series
\[
 \mathcal L_\infty(\kappa,\beta)
 :=
 \frac{2}{\pi}
 \sum_{m\in\mathbb Z}
 \widehat{\mathfrak a}(m)e^{2im\beta}e^{i\tau_\kappa m^3}
 =
 \frac{2}{\pi}\bigl(e^{-\tau_\kappa\partial_\theta^3}\mathfrak a\bigr)(2\beta),
\]
where $\tau_\kappa=\frac{16}{3}\kappa^3$. We must bound
\[
 \K_d\bigl(\kappa d^{-1/3},\beta\bigr)-\mathcal L_\infty(\kappa,\beta).
\]
We split the Fourier modes into low and high frequencies.

\medskip
\noindent\textbf{Step 1: choose the cutoff.}
Pick
\[
 A>\frac{1}{3\rho},
 \qquad
 N_d:=\lfloor A\log d\rfloor.
\]
Since $N_d d^{-1/3}\to 0$, there exists $d_1=d_1(\kappa_0)$ such that for all $d\ge d_1$, all $\abs m\le N_d$, and all $\abs\kappa\le\kappa_0$,
\[
 \abs{m\kappa d^{-1/3}}\le \frac1{16}.
\]

\medskip
\noindent\textbf{Step 2: low-frequency expansion.}
For $\abs m\le N_d$, set
\[
 z:=m\kappa d^{-1/3}.
\]
By Lemma~\ref{lem:lambda-expansion},
\[
 \log\Lambda(z)=\frac{16i}{3}z^3-16z^4+r(z),
 \qquad
 \abs{r(z)}\le 500\abs{z}^5.
\]
Multiplying by $d$ gives
\[
 d\log\Lambda(z)
 =i\tau_\kappa m^3-16\kappa^4m^4d^{-1/3}+\varepsilon_{m,d},
\]
where
\[
 \abs{\varepsilon_{m,d}}
 \le 500\abs{\kappa}^5\abs m^5d^{-2/3}
 \le 500\kappa_0^5\abs m^5d^{-2/3}.
\]
Define
\[
 \delta_{m,d}:=-16\kappa^4m^4d^{-1/3}+\varepsilon_{m,d}.
\]
Then
\[
 \Lambda\bigl(m\kappa d^{-1/3}\bigr)^d
 =e^{i\tau_\kappa m^3}e^{\delta_{m,d}}.
\]
Since $\abs m\le A\log d$,
\[
 \abs{\delta_{m,d}}
 \le
 16\kappa_0^4A^4(\log d)^4d^{-1/3}
 +
 500\kappa_0^5A^5(\log d)^5d^{-2/3}.
\]
The right-hand side tends to $0$, so there exists $d_2=d_2(\rho,\kappa_0)$ such that for all $d\ge d_2$,
\[
 \abs{\delta_{m,d}}\le 1
 \qquad (\abs m\le N_d,\ \abs\kappa\le\kappa_0).
\]
Therefore
\[
 \abs{e^{\delta_{m,d}}-1}\le 2\abs{\delta_{m,d}}.
\]
So
\[
 \left|
 \Lambda\bigl(m\kappa d^{-1/3}\bigr)^d-e^{i\tau_\kappa m^3}
 \right|
 \le
 32\kappa_0^4\abs m^4d^{-1/3}
 +
 1000\kappa_0^5\abs m^5d^{-2/3}.
\]
Hence the low-frequency contribution is bounded by
\begin{align*}
\frac{2}{\pi}
\sum_{\abs m\le N_d}
\abs{\widehat{\mathfrak a}(m)}
\left|
 \Lambda\bigl(m\kappa d^{-1/3}\bigr)^d-e^{i\tau_\kappa m^3}
\right|
&\le
\frac{64\kappa_0^4}{\pi}d^{-1/3}
\sum_{m\in\mathbb Z}\abs{\widehat{\mathfrak a}(m)}\abs m^4 \\
&\qquad+
\frac{2000\kappa_0^5}{\pi}d^{-2/3}
\sum_{m\in\mathbb Z}\abs{\widehat{\mathfrak a}(m)}\abs m^5.
\end{align*}
By Lemma~\ref{lem:fourier-decay}, both sums are finite. Therefore the low frequency error $E_L$ is bounded as follows:
\[
 E_L\le C_1d^{-1/3}.
\]

\medskip
\noindent\textbf{Step 3: high-frequency tail.}
For real $z$,
\[
 \abs{\Lambda(z)}=\frac{1}{\abs{1-4iz-8z^2}^{1/2}}.
\]
A direct computation gives
\[
 \abs{1-4iz-8z^2}^2=(1-8z^2)^2+16z^2=1+64z^4,
\]
so
\[
 \abs{\Lambda(z)}=(1+64z^4)^{-1/4}\le 1.
\]
Thus
\[
 \sum_{\abs m>N_d}
 \abs{\widehat{\mathfrak a}(m)}\abs{\Lambda(m\kappa d^{-1/3})}^d
 \le
 \sum_{\abs m>N_d}\abs{\widehat{\mathfrak a}(m)}.
\]
Also
\[
 \sum_{\abs m>N_d}\abs{\widehat{\mathfrak a}(m)}\abs{e^{i\tau_\kappa m^3}}
 =
 \sum_{\abs m>N_d}\abs{\widehat{\mathfrak a}(m)}.
\]
Hence the high frequency error $E_H$ is bounded as follows:
\[
 E_H
 \le
 \frac{4}{\pi}\sum_{\abs m>N_d}\abs{\widehat{\mathfrak a}(m)}.
\]
By Lemma~\ref{lem:fourier-decay},
\[
 \sum_{\abs m>N_d}\abs{\widehat{\mathfrak a}(m)}
 \le C_\rho e^{-\rho N_d}
 \le C_\rho d^{-A\rho}.
\]
Since $A\rho>1/3$,
\[
 E_H\le C_2d^{-1/3}.
\]

\medskip
\noindent\textbf{Step 4: combine the bounds.}
For $d\ge d_0:=\max\{d_1,d_2\}$,
\[
 \left|
 \K_d\bigl(\kappa d^{-1/3},\beta\bigr)-\mathcal L_\infty(\kappa,\beta)
 \right|
 \le
 E_L+E_H
 \le
 (C_1+C_2)d^{-1/3}.
\]
This proves \eqref{eq:main-estimate}. \qedhere
\end{proof}

\subsection{Explicit bounds on K\"onig's bilinear  form}

We are finally ready to instantiate the methodology above and get explicit bounds as in Theorem~\ref{th:konignain}. 

Because $\mathfrak a(\theta)=\arsinh(\cos\theta)$ is even and satisfies $\mathfrak a(\theta+\pi)=-\mathfrak a(\theta)$, its Fourier series contains only odd cosine modes:
\[
 \mathfrak a(\theta)=\sum_{\ell\ge0} b_{2\ell+1}\cos\bigl((2\ell+1)\theta\bigr),
\]
where
\[
 b_{2\ell+1}:=\frac{1}{\pi}\int_0^{2\pi}\mathfrak a(\theta)\cos\bigl((2\ell+1)\theta\bigr)\,d\theta,
 \qquad
 \abs{b_{2\ell+1}}\le C_\rho e^{-\rho(2\ell+1)}.
\]
From Theorem~\ref{thm:main} we immediately get the following explicit form of the limit.

\begin{corollary}[Cosine-series form]\label{cor:cosine}
For every fixed $\kappa,\beta\in\R$,
\[
 \lim_{d\to\infty}\K_d\bigl(\kappa d^{-1/3},\beta\bigr)
 =
 \frac{2}{\pi}
 \sum_{\ell\ge0}
 b_{2\ell+1}
 \cos\left(
 2(2\ell+1)\beta+\frac{16}{3}(2\ell+1)^3\kappa^3
 \right).
\]
\end{corollary}

Optimizing over $\kappa$ and $\beta$ yields 
\[
 \lim_{d\to\infty}\K_d\bigl(\kappa d^{-1/3},\beta\bigr)
 \geq 0.59357...,
\]
proving Theorem~\ref{th:konignain}. The exact formula \eqref{eq:exact-reduction} gives a direct numerical procedure: for each fixed auxiliary dimension $d$, truncate the absolutely convergent Fourier sum and optimize over $(\alpha,\beta)$.  Figure~\ref{fig:koenig-dimension} shows the resulting normalized K\"onig values.

\input{figures/affine_radial_koenig}

The optimized values in the displayed range exceed $0.589$ and are still growing, compared with the hyperplane value $0.56110\ldots$.  The cubic-multiplier limit explains the large-$d$ mechanism: under the $d^{-1/3}$ scaling, the radial phase creates the multiplier $m\mapsto e^{i(16/3)\kappa^3m^3}$ appearing in \eqref{eq:main-estimate}. The main point for the present section is the resulting explicit asymptotic source of non-hyperplane improvements for K\"onig's bilinear form.


%% file: figures/affine_radial_koenig.tex
\begin{figure}[t]
\centering
\includegraphics[width=0.82\textwidth]{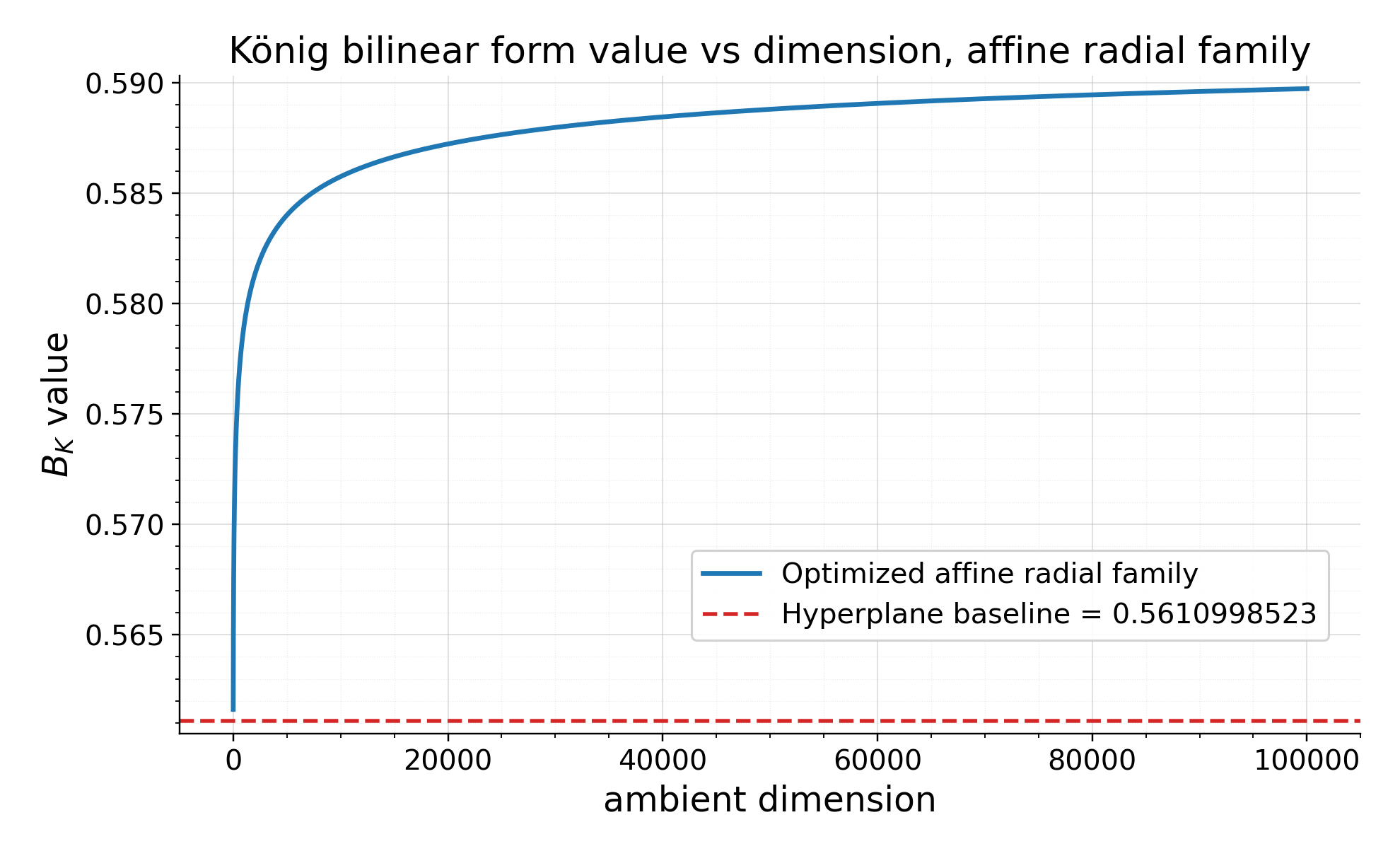}
\caption{Normalized K\"onig's bilinear form $\mathcal B_{\mathrm K}$ for affine radial rotating-hyperplane partitions as a function of the ambient dimension dimension $d+2$. The dashed horizontal line is the hyperplane benchmark value of $\tfrac{2}{\pi} \log(1+\sqrt2) =0.5610998523 \ldots$.}
\label{fig:koenig-dimension}
\end{figure}

%% file: paper/discussion.tex
\section{Discussion and Related Work}

\paragraph{Open Challenges}
Asymptotic results for Grothendieck's constant are of interest, with potential avenues being extensions of our results for the K\"onig bilinear form.
Improving the \(K_G\) lower bound is also of interest (see~{\cite{Davie1984LowerBoundKG, Reeds1991LowerBoundKG,heilman2026lowerboundgrothendiecksconstant, jones2026grothendieckconstantstrictlylarger}}), and may be amenable to AI-search, as common strategies involve the construction and analysis of particular semidefinite programs, akin to our efforts in partition search.


\paragraph{AI in Mathematics}
Much of this work was made possible due to tooling such as GPT-5's Deep Research feature~\cite{openai:2025:deepresearch} and Claude Code's AI agentic capabilities~\cite{anthropic:2026:opus4.7}.
More broadly, AI is increasingly used in mathematics research, with notable recent advances made in the disproof of the unit distance conjecture~\cite{openai2026planarpointsets} and solutions to a number of Erd\H{o}s problems~\cite{alexeev2026primitivesetsvonmangoldt,erdos,diez2026mathematical, tsoukalas2026advancingmathematicsresearchaidriven}.
Benchmarking mathematical performance is also of interest, with a focus on problems and tasks relevant to research~\cite{abouzaid2026firstproof,son2026soohak}.
Concurrently, there is ongoing development of AI agents and systems for mathematics research, with notable recent examples including~\cite{Trinh2024-alphageometry-1,Hubert2025,feng2026aletheiatacklesfirstproofautonomously}.

%% file: paper/acknowledgements.tex
\paragraph{Acknowledgements}

This research was partially funded by the NSF AI Institute for Foundations of Machine Learning (IFML), NSF EnCORE: Institute for Emerging CORE Methods in Data Science Award \#2217033, NSF award \#2403211, a Guggenheim Fellowship, and the Renaissance Philanthropy AI for Math Fund. We would like to thank Sabrina Reguyal for reading the manuscript and offering helpful comments and suggestions on the writing.


%% file: paper/appendix.tex
\newpage
\appendix

\section{Krivine Preprocessing and Mixed Rounding}
\label[appendix]{app:krivine-preprocessing}

This appendix records the preprocessing and projection arguments behind
Krivine rounding schemes.
The argument is the standard rounding construction
~\cite{krivine1977constante, braverman2011grothendieckconstantstrictlysmaller}.

For a real matrix \(A=(a_{ij})\in\mathbb R^{m\times n}\), write
\[
    \operatorname{SDP}(A)
    :=
    \sup_{\|x_i\|=\|y_j\|=1}
    \left|
        \sum_{i=1}^m\sum_{j=1}^n a_{ij}\langle x_i,y_j\rangle
    \right|
\]
and
\[
    \operatorname{OPT}(A)
    :=
    \max_{\varepsilon_i,\delta_j\in\{\pm 1\}}
    \left|
        \sum_{i=1}^m\sum_{j=1}^n a_{ij}\varepsilon_i\delta_j
    \right|.
\]
Thus proving \(\operatorname{SDP}(A)\le K\operatorname{OPT}(A)\) for all
\(A\) proves \(K_G\le K\).

\begin{theorem}[Krivine preprocessing and projection]
\label{thm:appendix-pure-krivine-rounding}
Let \(f,g:\mathbb R^k\to\{-1,1\}\) be odd measurable functions, and let
\(H=H_{f,g}\) be the arcsine-normalized correlation function
\[
    H(t)
    =
    \frac{\pi}{2}
    \mathbb E\left[
        f\left(G_1\right)
        g\left(tG_1+\sqrt{1-t^2}G_2\right)
    \right],
    \qquad -1<t<1,
\]
where \(G_1,G_2\) are independent standard Gaussian vectors in
\(\mathbb R^k\). Suppose that \(H\) has a local inverse at the origin,
\[
    H^{-1}(\zeta)=\sum_{r\ge 1} a_r\zeta^r,
\]
which is holomorphic on a neighborhood of \(\gamma\overline{\mathbb D}\),
and suppose that
\[
    \sum_{r\ge 1}|a_r|\gamma^r\le 1.
\]
Then
\[
    K_G\le \frac{\pi}{2\gamma}.
\]
\end{theorem}

\begin{proof}
It is enough to prove the conclusion with any \(0<\gamma'<\gamma\) in
place of \(\gamma\), and then let \(\gamma'\uparrow\gamma\). We therefore
fix such a \(\gamma'\). Put
\[
    M_{\gamma'}:=\sum_{r\ge 1}|a_r|(\gamma')^r < 1.
\]
Let \(A=(a_{ij})\in\mathbb R^{m\times n}\). Choose unit vectors
\(x_1,\ldots,x_m,y_1,\ldots,y_n\) such that
\[
    \sum_{i,j} a_{ij}\langle x_i,y_j\rangle
\]
is arbitrarily close to \(\operatorname{SDP}(A)\); if necessary, replace all
\(y_j\) by \(-y_j\) so that this quantity is nonnegative. We suppress this
arbitrarily small error, since it disappears at the end.

Let
\[
    \mathcal K
    :=
    \left(\bigoplus_{r\ge 1} \mathcal H^{\otimes r}\right)
    \oplus \mathbb R e_L\oplus \mathbb R e_R,
\]
where \(\mathcal H\) is the Hilbert space containing the original vectors,
and where \(e_L,e_R\) are unit vectors orthogonal to each other and to all
tensor summands. Define
\[
    I(x)
    :=
    \bigoplus_{r\ge 1}
        \sqrt{|a_r|}\,(\gamma')^{r/2}x^{\otimes r}
    \oplus \sqrt{1-M_{\gamma'}}\,e_L
    \oplus 0
\]
and
\[
    J(y)
    :=
    \bigoplus_{r\ge 1}
        \operatorname{sgn}(a_r)\sqrt{|a_r|}\,(\gamma')^{r/2}y^{\otimes r}
    \oplus 0
    \oplus \sqrt{1-M_{\gamma'}}\,e_R,
\]
with the convention that the summand corresponding to \(a_r=0\) is zero.
Then \(I(x)\) and \(J(y)\) are unit vectors whenever \(x,y\) are unit
vectors. Moreover,
\[
    \langle I(x),J(y)\rangle
    =
    \sum_{r\ge 1} a_r(\gamma')^r\langle x,y\rangle^r
    =
    H^{-1}\bigl(\gamma'\langle x,y\rangle\bigr).
\]
The first equality uses
\[
    \langle x^{\otimes r},y^{\otimes r}\rangle
    =
    \langle x,y\rangle^r,
\]
and the second equality uses the Taylor expansion of the inverse on
\(\gamma'\mathbb D\).

Set
\[
    u_i:=I(x_i),
    \qquad
    v_j:=J(y_j).
\]
The finitely many vectors \(u_i,v_j\) span a finite-dimensional subspace of
\(\mathcal K\). Identify this subspace with \(\mathbb R^N\). Let
\(\Gamma:\mathbb R^N\to\mathbb R^k\) be a random Gaussian linear map, i.e.
a \(k\times N\) matrix whose entries are independent standard Gaussian
random variables. Define random signs
\[
    \varepsilon_i
    :=
    f\left(\Gamma u_i\right),
    \qquad
    \delta_j
    :=
    g\left(\Gamma v_j\right).
\]
For each pair \((i,j)\), the vectors \(\Gamma u_i\) and \(\Gamma v_j\) are
standard Gaussian vectors in \(\mathbb R^k\) with coordinatewise correlation
\(\langle u_i,v_j\rangle\). Hence, by the definition of \(H\),
\[
    \mathbb E[\varepsilon_i\delta_j]
    =
    \frac{2}{\pi}
    H(\langle u_i,v_j\rangle).
\]
Using the preprocessing identity,
\[
    H(\langle u_i,v_j\rangle)
    =
    H\left(H^{-1}\bigl(\gamma'\langle x_i,y_j\rangle\bigr)\right)
    =
    \gamma'\langle x_i,y_j\rangle.
\]
Therefore
\[
\begin{aligned}
    \mathbb E\left[\sum_{i,j} a_{ij}\varepsilon_i\delta_j\right]
    &=
    \frac{2}{\pi}\sum_{i,j}a_{ij}
        H(\langle u_i,v_j\rangle) \\
    &=
    \frac{2\gamma'}{\pi}
    \sum_{i,j}a_{ij}\langle x_i,y_j\rangle.
\end{aligned}
\]
Since the maximum over deterministic signs is at least the expectation of
the randomized signs, we get
\[
    \operatorname{OPT}(A)
    \ge
    \frac{2\gamma'}{\pi}\operatorname{SDP}(A).
\]
Thus
\[
    \operatorname{SDP}(A)
    \le
    \frac{\pi}{2\gamma'}\operatorname{OPT}(A).
\]
Letting \(\gamma'\uparrow\gamma\) gives
\[
    \operatorname{SDP}(A)
    \le
    \frac{\pi}{2\gamma}\operatorname{OPT}(A).
\]
Since \(A\) was arbitrary, \(K_G\le \pi/(2\gamma)\).
\end{proof}

\begin{corollary}[Classical Krivine hyperplane bound]
\label{cor:appendix-classical-krivine}
Let \(h:\mathbb R\to\{-1,1\}\) be \(h(x)=\operatorname{sgn}(x)\). Then the
hyperplane scheme \((h,h)\) gives
\[
    K_G\le \frac{\pi}{2\log(1+\sqrt 2)}.
\]
\end{corollary}

\begin{proof}
For the hyperplane scheme, Grothendieck's identity gives
\[
    \mathbb E[\operatorname{sgn}(X)\operatorname{sgn}(Y)]
    =
    \frac{2}{\pi}\arcsin(t)
\]
whenever \(X,Y\) are standard Gaussians with correlation \(t\). In our
normalization,
\[
    H_{h,h}(t)=\arcsin(t).
\]
Thus
\[
    H_{h,h}^{-1}(\zeta)=\sin \zeta
    =
    \sum_{q\ge 0}\frac{(-1)^q}{(2q+1)!}\zeta^{2q+1}.
\]
Let
\[
    \rho_*:=\operatorname{arsinh}(1)=\log(1+\sqrt 2).
\]
Then
\[
    \sum_{q\ge 0}
    \frac{\rho_*^{2q+1}}{(2q+1)!}
    =
    \sinh(\rho_*)
    =
    1.
\]
Theorem~\ref{thm:appendix-pure-krivine-rounding} therefore applies with
\(\gamma=\rho_*\), giving
\[
    K_G\le \frac{\pi}{2\rho_*}
    =
    \frac{\pi}{2\log(1+\sqrt 2)}.
\]
\end{proof}

\begin{theorem}[Mixed Krivine preprocessing and projection]
\label{thm:appendix-mixed-krivine-rounding}
Let \(L\in\mathbb N\). For each \(\ell\in\{1,\ldots,L\}\), let
\(f_\ell,g_\ell:\mathbb R^{k_\ell}\to\{-1,1\}\) be odd measurable functions,
and let \(H_\ell\) be the corresponding arcsine-normalized correlation
function:
\[
    H_\ell(t)
    =
    \frac{\pi}{2}
    \mathbb E\left[
        f_\ell\left(G_1\right)
        g_\ell\left(tG_1+\sqrt{1-t^2}G_2\right)
    \right].
\]
Let \(\lambda_1,\ldots,\lambda_L\ge 0\) satisfy
\(\sum_{\ell=1}^L\lambda_\ell=1\), and define the mixed function
\[
    H_\lambda(t):=\sum_{\ell=1}^L\lambda_\ell H_\ell(t).
\]
Suppose that \(H_\lambda\) has a local inverse at the origin,
\[
    H_\lambda^{-1}(\zeta)=\sum_{r\ge 1} a_r\zeta^r,
\]
which is holomorphic on a neighborhood of \(\gamma\overline{\mathbb D}\),
and suppose that
\[
    M_\lambda(\gamma)
    :=
    \sum_{r\ge 1}|a_r|\gamma^r
    \le 1.
\]
Then
\[
    K_G\le \frac{\pi}{2\gamma}.
\]
\end{theorem}

\begin{proof}
The preprocessing is identical to the proof of
Theorem~\ref{thm:appendix-pure-krivine-rounding}, but using the inverse
coefficients of \(H_\lambda^{-1}\). As before, first fix
\(0<\gamma'<\gamma\), and set
\[
    M_{\lambda,\gamma'}
    :=
    \sum_{r\ge 1}|a_r|(\gamma')^r
    <1.
\]
For unit vectors \(x,y\), define
\[
    I(x)
    :=
    \bigoplus_{r\ge 1}
        \sqrt{|a_r|}\,(\gamma')^{r/2}x^{\otimes r}
    \oplus \sqrt{1-M_{\lambda,\gamma'}}\,e_L
    \oplus 0
\]
and
\[
    J(y)
    :=
    \bigoplus_{r\ge 1}
        \operatorname{sgn}(a_r)\sqrt{|a_r|}\,(\gamma')^{r/2}y^{\otimes r}
    \oplus 0
    \oplus \sqrt{1-M_{\lambda,\gamma'}}\,e_R.
\]
Then \(I(x)\) and \(J(y)\) are unit vectors and
\[
    \langle I(x),J(y)\rangle
    =
    H_\lambda^{-1}\bigl(\gamma'\langle x,y\rangle\bigr).
\]

Now let \(A=(a_{ij})\) be arbitrary, and choose unit vectors
\(x_i,y_j\) that nearly attain \(\operatorname{SDP}(A)\), oriented so that
\(\sum_{i,j}a_{ij}\langle x_i,y_j\rangle\ge 0\). Put
\[
    u_i:=I(x_i),
    \qquad
    v_j:=J(y_j).
\]
Identify the span of the finitely many \(u_i,v_j\) with \(\mathbb R^N\).
For each \(\ell\), independently sample a Gaussian linear map
\[
    \Gamma_\ell:\mathbb R^N\to\mathbb R^{k_\ell}.
\]
Also sample a single global index \(\Lambda\in\{1,\ldots,L\}\), independent
of the Gaussian maps, with
\[
    \mathbb P[\Lambda=\ell]=\lambda_\ell.
\]
The word ``global'' is important: the same index \(\Lambda\) is used for
all \(i,j\). Define
\[
    \varepsilon_i
    :=
    f_\Lambda\left(\Gamma_\Lambda u_i\right),
    \qquad
    \delta_j
    :=
    g_\Lambda\left(\Gamma_\Lambda v_j\right).
\]
Conditioning on \(\Lambda=\ell\), the same rotation-invariance calculation
as in the pure case gives
\[
    \mathbb E[\varepsilon_i\delta_j\mid \Lambda=\ell]
    =
    \frac{2}{\pi}H_\ell(\langle u_i,v_j\rangle).
\]
Averaging over \(\Lambda\),
\[
    \mathbb E[\varepsilon_i\delta_j]
    =
    \frac{2}{\pi}
    \sum_{\ell=1}^L\lambda_\ell H_\ell(\langle u_i,v_j\rangle)
    =
    \frac{2}{\pi}H_\lambda(\langle u_i,v_j\rangle).
\]
Since
\[
    \langle u_i,v_j\rangle
    =
    H_\lambda^{-1}\bigl(\gamma'\langle x_i,y_j\rangle\bigr),
\]
we obtain
\[
    \mathbb E[\varepsilon_i\delta_j]
    =
    \frac{2\gamma'}{\pi}\langle x_i,y_j\rangle.
\]
Therefore
\[
\begin{aligned}
    \operatorname{OPT}(A)
    &\ge
    \mathbb E\left[
        \sum_{i,j}a_{ij}\varepsilon_i\delta_j
    \right] \\
    &=
    \frac{2\gamma'}{\pi}
    \sum_{i,j}a_{ij}\langle x_i,y_j\rangle.
\end{aligned}
\]
Taking the supremum over the original vectors gives
\[
    \operatorname{OPT}(A)
    \ge
    \frac{2\gamma'}{\pi}\operatorname{SDP}(A),
\]
and hence
\[
    \operatorname{SDP}(A)
    \le
    \frac{\pi}{2\gamma'}\operatorname{OPT}(A).
\]
Letting \(\gamma'\uparrow\gamma\) proves
\[
    \operatorname{SDP}(A)
    \le
    \frac{\pi}{2\gamma}\operatorname{OPT}(A).
\]
Since \(A\) was arbitrary, \(K_G\le \pi/(2\gamma)\).
\end{proof}

\begin{remark}
Theorem~\ref{thm:appendix-pure-krivine-rounding} is the special case
\(L=1\) of Theorem~\ref{thm:appendix-mixed-krivine-rounding}. We separated
the two statements because the pure case is the usual Krivine scheme, while
the mixed case emphasizes the randomized choice of a common scheme
\(\Lambda\) before projection.
\end{remark}

\section{Maximal Bounds}
\label[appendix]{app:maximal-bounds}

\begin{lemma}[A disk-to-strip bound for the sine map]
\label[lemma]{lemma:disk-to-strip-check}
Let \(R=1.1\). For all \(u,v\in\mathbb{R}\) with
\[
u^2+v^2\le R^2,
\]
one has
\[
|\sin u|\cosh v<0.946.
\]
In particular, if \(w=u+iv\) and \(|w|\le 1.1\), then
\[
|\Re(\sin w)|<0.946<1.
\]
\end{lemma}

\begin{proof}
Since \(R<\pi/2\), by symmetry we may assume \(0\le u\le R\) and
\(v\ge 0\). For fixed \(u\), the function \(v\mapsto \cosh v\) is
increasing for \(v\ge 0\). Hence the maximum over \(u^2+v^2\le R^2\)
is attained when
\[
v=\sqrt{R^2-u^2}.
\]
It therefore suffices to maximize
\[
\phi(u):=\sin u\cosh\sqrt{R^2-u^2},
\qquad 0\le u\le R.
\]

Put
\[
s(u):=\sqrt{R^2-u^2}.
\]
For \(0<u<R\), differentiation gives
\[
\phi'(u)
=
\sin u\cosh s(u)
\left(
\cot u-\frac{u}{s(u)}\tanh s(u)
\right).
\]
Define
\[
h(u):=\cot u-\frac{u}{s(u)}\tanh s(u).
\]
We claim that \(h\) is strictly decreasing on \((0,R)\). The function
\(u\mapsto \cot u\) is strictly decreasing. Also
\[
A(s):=\frac{\tanh s}{s}
\]
is strictly decreasing for \(s>0\), since
\[
A'(s)=\frac{s\operatorname{sech}^2s-\tanh s}{s^2}<0,
\]
the last inequality being equivalent to
\[
\sinh s\cosh s-s>0
\]
for \(s>0\). Since \(s(u)\) is decreasing, the function \(A(s(u))\)
is increasing, and hence \(uA(s(u))\) is increasing. Therefore
\[
h(u)=\cot u-uA(s(u))
\]
is strictly decreasing.

It remains to locate the unique critical point. Direct outward-rounded
Taylor bounds give
\[
h(0.898)>0.796708-0.793956>0
\]
and
\[
h(0.900)<0.793553-0.796524<0.
\]
For example, these inequalities follow from the alternating Taylor
bounds for \(\sin\) and \(\cos\), and from writing
\[
\frac{u}{s}\tanh s
=
u\,\frac{\sinh s}{s\cosh s},
\]
together with Taylor bounds for \(\sinh s/s\) and \(\cosh s\).
Since \(h\) is strictly decreasing, the unique maximizer \(u_*\) of
\(\phi\) lies in
\[
0.898<u_*<0.900.
\]

Now take \(c=0.899\). On the interval \([0.898,0.900]\), another
differentiation gives
\[
\phi''(u)
=
\sin u\cosh s
\left(\frac{u^2}{s^2}-1\right)
-2\frac{u}{s}\cos u\sinh s
-\frac{R^2}{s^3}\sin u\sinh s,
\]
where \(s=s(u)\). Throughout this interval,
\[
s^2\ge 0.4,\qquad s>0.632,\qquad s<0.636,
\]
and
\[
\cosh s<1.21,\qquad \sinh s<0.68.
\]
Substituting these bounds into the displayed formula for \(\phi''\)
gives
\[
|\phi''(u)|<7
\qquad (0.898\le u\le 0.900).
\]

Since \(u_*\in[0.898,0.900]\), we have \(|u_*-c|\le 10^{-3}\). Taylor's
theorem gives
\[
\phi(u_*)\le \phi(c)+\frac{7}{2}\cdot 10^{-6}.
\]
Next,
\[
s(c)^2=1.21-0.899^2=0.401799<0.6339^2,
\]
so
\[
\cosh s(c)<\cosh(0.6339)<1.20774.
\]
Also, by the alternating Taylor expansion of \(\sin\),
\[
\sin(0.899)
<
0.899-\frac{0.899^3}{6}+\frac{0.899^5}{120}
<0.7828.
\]
Therefore
\[
\phi(c)<0.7828\cdot 1.20774<0.945419.
\]
Consequently,
\[
\max_{u^2+v^2\le 1.1^2}|\sin u|\cosh v
=
\phi(u_*)
<
0.945419+0.0000035
<
0.946.
\]
This proves the claim.
\end{proof}

\begin{lemma}[Cosine-ratio bound in the sine coordinate]
\label[lemma]{lemma:cosine-ratio-check}
For every $w\in\C$ with $|w|\le 1.1$,
\[
 \frac{\pi}{2}\frac{|\cos w|^2}{1-(\Re \sin w)^2}<14.19.
\]
\end{lemma}

\begin{proof}
Write $w=u+iv$ and put $\rho=1.1$. We need to maximize over
\[
 u^2+v^2\le \rho^2.
\]
We have
\[
 \Re(\sin w)=\sin u\cosh v
\]
and
\[
 |\cos w|^2
 =
 |\cos u\cosh v-i\sin u\sinh v|^2
 =
 \cos^2u+\sinh^2v.
\]
Hence it is enough to bound
\[
 Q(u,v):=
 \frac{\cos^2u+\sinh^2v}{1-\sin^2u\cosh^2v}
\]
on the disk $u^2+v^2\le \rho^2$. By the previous bound
\[
 |\sin u|\cosh v<1
\]
on this disk, the denominator is positive.

Set
\[
 x:=\sin^2u,\qquad y:=\sinh^2v.
\]
Then
\[
 Q(u,v)=\frac{1-x+y}{1-x-xy}.
\]
For fixed $u$, equivalently fixed $x$, we compute
\[
 \frac{\partial}{\partial y}
 \frac{1-x+y}{1-x-xy}
 =
 \frac{1-x^2}{(1-x-xy)^2}>0.
\]
Therefore, for each fixed $u$, the maximum occurs when $|v|$ is as large as possible, namely
\[
 v=\sqrt{\rho^2-u^2}.
\]
By evenness in $u$, it remains to maximize the one-variable function
\[
 \Phi(u):=
 \frac{\cos^2u+\sinh^2s(u)}
 {1-\sin^2u\cosh^2s(u)},
 \qquad
 s(u):=\sqrt{\rho^2-u^2},
 \qquad
 0\le u\le \rho.
\]

For $0<u<\rho$, differentiating gives
\[
 \Phi'(u)
 =
 \frac{\sinh(2s)\cos u}
 {\bigl(1-\sin^2u\cosh^2s\bigr)^2}
 \left[
 \sin u\,\tanh s\,(1+\cosh^2s)
 -
 \frac{u}{s}\cos u\,(1+\sin^2u)
 \right],
\]
where $s=s(u)$. The prefactor is positive, since $\rho<\pi/2$. Thus the sign of
$\Phi'(u)$ is the sign of $A(u)-1$, where
\[
 A(u):=
 \frac{s(u)\tan u\,\tanh s(u)\,(1+\cosh^2s(u))}
 {u(1+\sin^2u)}.
\]

We claim that $A$ is strictly decreasing on $(0,\rho)$. Let
\[
 B(s):=s\tanh s(1+\cosh^2s).
\]
Then
\[
 \frac{d}{du}\log A(u)
 =
 \left(\frac{1}{\sin u\cos u}-\frac1u\right)
 -\frac{\sin(2u)}{1+\sin^2u}
 -\frac{u}{s}\frac{B'(s)}{B(s)}.
\]
Moreover,
\[
 \frac{B'(s)}{B(s)}
 =
 \frac1s+\frac{1}{\sinh s\cosh s}
 +
 \frac{2\sinh s\cosh s}{1+\cosh^2s}.
\]
Since
\[
 \frac{2\sinh s\cosh s}{1+\cosh^2s}\ge \tanh s
\]
and
\[
 \frac{1}{\sinh s\cosh s}+\tanh s=\coth s>\frac1s,
\]
we get
\[
 \frac{B'(s)}{B(s)}>\frac2s.
\]
On the other hand, for $0<u<\rho$,
\[
 \sin u\cos u=\frac12\sin(2u)>u-\frac23u^3.
\]
Therefore
\[
 \frac{1}{\sin u\cos u}-\frac1u
 =
 \frac{u-\sin u\cos u}{u\sin u\cos u}
 <
 \frac{2u}{3-2u^2}
 <
 \frac{2u}{\rho^2-u^2}
 =
 \frac{2u}{s^2}.
\]
Combining these estimates,
\[
 \frac{d}{du}\log A(u)
 <
 \frac{2u}{s^2}
 -
 \frac{\sin(2u)}{1+\sin^2u}
 -
 \frac{2u}{s^2}
 <0.
\]
Thus $A$ is strictly decreasing. Since $A(u)>1$ for small $u>0$ and
$A(u)\to0$ as $u\uparrow\rho$, the function $\Phi$ increases up to the unique
point where $A(u)=1$ and decreases afterward.

It remains only to certify the location and value of this maximum. A direct outward-rounded
interval computation gives
\[
 A(0.8289364)>1.0000019,
 \qquad
 A(0.8289372)<0.9999992.
\]
Hence the unique maximizer of $\Phi$ lies in
\[
 I=[0.8289364,\,0.8289372].
\]
Evaluating $\Phi$ on this interval, again with outward-rounded interval arithmetic, gives
\[
 \Phi(I)\subset [\,9.0332175,\;9.03345993\,].
\]
Consequently,
\[
 \sup_{|w|\le 1.1}
 \frac{\pi}{2}\frac{|\cos w|^2}{1-(\Re\sin w)^2}
 =
 \frac{\pi}{2}\sup_{0\le u\le \rho}\Phi(u)
 \le
 \frac{\pi}{2}\cdot 9.03345993
 <
 14.18973
 <
 14.19.
\]
This proves the claim.
\end{proof}

\section{Coefficient Vectors and Mixing Parameter.}
\label[appendix]{app:explicit-mixture-coefficients}

The coefficient vectors and hyperplane weights by degree are listed below. 

\begin{center}
\renewcommand{\arraystretch}{1.3}
\small
\setlength{\tabcolsep}{4pt}
\begin{tabular}{|>{\centering\arraybackslash}p{1.2cm}|l|}
\hline
$d$ & Coefficients (rounded; full version in repository) \\
\hline
$3$ &
  $a^{(3)} = (+0.001558,\; +0.076353)$ \\
    & $b^{(3)} = (-0.000315,\; -0.076352)$ \\
    & $p_3 = 0.287269$ \\ \hline 
$5$ &
  $a^{(5)} = (+0.004898,\; +0.073363,\; +0.008537)$ \\
    & $b^{(5)} = (+0.002490,\; -0.072991,\; -0.008345)$ \\
    & $p_5 = 0.275000$ \\ \hline 
$7$ &
  $a^{(7)} = (+0.091892,\; -0.068852,\; -0.006028,\; -0.003536)$ \\
    & $b^{(7)} = (+0.096058,\; +0.083071,\; +0.025794,\; +0.014605)$ \\
    & $p_7 = 0.275923$ \\ \hline 
$9$ &
  $a^{(9)} = (-0.069245,\; -0.083730,\; -0.034509,\; -0.030311,\; -0.010654)$ \\
    & $b^{(9)} = (-0.064714,\; +0.068967,\; +0.006236,\; +0.001519,\; -0.002667)$ \\
    & $p_9 = 0.273360$ \\
\hline
\end{tabular}
\end{center}

\newpage 
\section{Parameter Table for Certifying the Rounding Schemes}
\label[appendix]{app:parameter-table}

All five certificates share the constants 
\[
  M = 0.98,\qquad R = 0.975,\qquad r_0 = 1.1,\qquad B = 15.3,
  \qquad \rho_* = \log(1+\sqrt2)=0.8813735870\ldots
\]

\begin{table}[ht]
\centering
\renewcommand{\arraystretch}{1.25}
\setlength{\tabcolsep}{3pt}
\scriptsize
\begin{adjustbox}{max width=\textwidth}
\begin{tabular}{@{}l c c c c c@{}}
\toprule
 & \textbf{\cref{thm:pure-h3} ($\mathrm{He}_3$)} & \textbf{deg 3} & \textbf{deg 5} & \textbf{deg 7} & \textbf{deg 9} \\
\midrule
$N_0$ (cond.\ 2 head deg.) & $144$ & $106$ & $106$ & $106$ & $106$ \\
$N_1$ (cond.\ 3\&4 head deg.) & $220$ & $196$ & $196$ & $180$ & $192$ \\
$\gamma-\rho_*$ & $4.8\cdot10^{-6}$ & $1.403\cdot10^{-5}$ & $2.755\cdot10^{-5}$ & $2.945\cdot10^{-5}$ & $2.987\cdot10^{-5}$ \\
\midrule
\multicolumn{6}{@{}l}{\emph{Condition 2 (Rouch\'e inverse disk), at $M=0.98$:}}\\
$\displaystyle\sum_{n\le N_0}|d_n|M^n$
 & $<0.00451411$ & $<0.00414862$ & $<0.00359494$ & $<0.00356722$ & $<0.00346558$ \\
$\displaystyle B\frac{(M/r_0)^{N_0+1}}{1-M/r_0}$
 & $<7.4601\cdot10^{-6}$ & $<6.0128\cdot10^{-4}$ & $<6.0128\cdot10^{-4}$ & $<6.0128\cdot10^{-4}$ & $<6.0128\cdot10^{-4}$ \\
\midrule
\multicolumn{6}{@{}l}{\emph{Conditions 3 \& 4 (admissibility at $\gamma$):}}\\
$\displaystyle\sum_{n\le N_1}|a_n|\gamma^n$
 & $\le 0.999999957755$ & $\le 0.999999957887$ & $\le 0.999999966102$ & $\le 0.999999857455$ & $\le 0.999999950702$ \\
$T=1-\text{head}$
 & $4.224\cdot10^{-8}$ & $4.211\cdot10^{-8}$ & $3.389\cdot10^{-8}$ & $1.425\cdot10^{-7}$ & $4.929\cdot10^{-8}$ \\
$\displaystyle \sinh(M)\frac{(\gamma/R)^{N_1+1}}{1-\gamma/R}\ (\le T)$
 & $\le 2.44\cdot10^{-9}$ & $\le 2.757\cdot10^{-8}$ & $\le 2.765\cdot10^{-8}$ & $\le 1.391\cdot10^{-7}$ & $\le 4.143\cdot10^{-8}$ \\
\midrule
$K_G\le$
 & $1.782204273$ & $1.782185609$ & $1.782158272$ & $1.782154430$ & $1.782153581$ \\
improvement over $\tfrac{\pi}{2\rho_*}$
 & $\ge 9.705\cdot10^{-6}$ & $\ge 2.836\cdot10^{-5}$ & $\ge 5.570\cdot10^{-5}$ & $\ge 5.954\cdot10^{-5}$ & $\ge 6.039\cdot10^{-5}$ \\
\bottomrule
\end{tabular}
\end{adjustbox}
\caption{Certified parameters for the five rounding schemes. Shared constants:
$M=0.98$, $R=0.975$, $r_0=1.1$, $B=15.3$, $\rho_*=\log(1+\sqrt2)$. For each scheme
$\gamma=\rho_*+(\gamma-\rho_*)$. Condition~2 truncates the perturbation series at
degree $N_0$ and bounds the tail $n>N_0$ analytically; conditions 3 \& 4 truncate
the inverse-majorant at $N_1$ with Cauchy remainder $\le T=1-\text{head}$. All
head/tail/remainder/$K_G$ entries are rigorous upper bounds; the improvement is a
rigorous lower bound.}
\label{tab:scheme-parameters}
\end{table}